\documentclass{article}

\usepackage[preprint]{neurips_2026}

\usepackage[utf8]{inputenc} 
\usepackage[T1]{fontenc}    
\usepackage{amsmath,amssymb,amsfonts}
\usepackage{amsthm}

\usepackage{graphicx}
\usepackage{algorithm}
\usepackage{algorithmic}
\usepackage{color}
\usepackage{bm}
\usepackage{xcolor}
\usepackage{multirow}
\usepackage{diagbox}
\usepackage{array,makecell,booktabs}
\usepackage{threeparttable}
\newcolumntype{M}[1]{>{\centering\arraybackslash}m{#1}}
\usepackage{adjustbox}
\usepackage[hidelinks]{hyperref}
\usepackage{listings}
\usepackage{matlab-prettifier}
\usepackage{bbding}
\usepackage{footmisc}
\usepackage{cleveref}
\usepackage{longtable}
\usepackage[misc]{ifsym}
\usepackage[T1]{fontenc}
\usepackage{todonotes}
\usepackage{enumitem}

\usepackage{subfigure}
\usepackage{wrapfig}

\newcommand{\oomit}[1]{}

\DeclareMathOperator{\argmin}{argmin}

\newtheorem{assumption}{Assumption}

\newtheorem{proposition}{Proposition}
\newtheorem{lemma}{Lemma}
\newtheorem{theorem}{Theorem}
\newtheorem{corollary}{Corollary}

\title{Cost-Aware Adaptive Conformal Inference for Runtime Assurance in Dynamic Environments}

\author{%
  Taoran Wu \\
  KLSS, Institute of Software, CAS\\
  University of Chinese Academy of Sciences\\
  Beijing, China \\
  \texttt{wutr@ios.ac.cn} \\
  \And
  Jingduo Pan \\
  KLSS, Institute of Software, CAS\\
  University of Chinese Academy of Sciences\\
  Beijing, China \\
  \texttt{panjd@ios.ac.cn} \\
  \And
  Luke Ong \\
  Nanyang Technological University\\
  Singapore \\
  \texttt{luke.ong@ntu.edu.sg} \\
   \And
  Bai Xue \thanks{Corresponding author}\\
  KLSS, Institute of Software, CAS\\
  University of Chinese Academy of Sciences\\
  Beijing, China \\
  \texttt{xuebai@ios.ac.cn} \\
}

\begin{document}

\maketitle 

\begin{abstract}
This paper addresses the problem of providing runtime assurance for systems operating online under unknown and potentially time-varying data distributions. We propose Cost-Aware Adaptive Conformal Inference (ACI), a novel framework that incorporates constraint violation costs directly into the conformal adaptation mechanism. Our key insight is that uncertainty margins should adapt not only to the frequency of constraint violations but also to their severity. We formalize this through a cost-aware loss function that couples the miscoverage indicator with violation costs. Unlike existing methods that regulate a single controlled metric, our approach provides a dual statistical guarantee: simultaneously bounding the long-run average violation frequencies (reliability) and cumulative violation cost (harm). By weighting prediction failures according to their severity, the algorithm enables the controller to respond proportionally to violation severity, expanding prediction sets aggressively when necessary while maintaining efficiency during nominal operation. We integrate Cost-Aware ACI into a robust control synthesis framework, creating a closed-loop system that balances task performance with runtime risk control without requiring explicit model knowledge. Experiments validate its effectiveness for online risk-aware controller synthesis.
\end{abstract}

\section{Introduction}
\label{sec:intro}
Deploying autonomous systems in safety-critical domains such as robotics and autonomous vehicles requires rigorous operational assurances. In formal verification, ``safety'' is typically defined as strict set invariance, which guarantees that the system never violates constraints \cite{ames2019control}. While this is a gold standard, establishing such absolute guarantees often presupposes accurate models and static environments. In real-world applications characterized by unknown dynamics and distribution shifts \cite{dixit2023adaptive}, strict invariance is often unattainable, as these factors preclude effective offline calibration. The challenge, therefore, lies in proving runtime assurance despite these uncertainties.    


To address this, we turn to the paradigm of runtime assurance and risk control. Rather than certifying total absence of failure, we aim to bound the frequency and severity of constraint violations online. Recent progress in distribution-free uncertainty quantification has positioned Conformal Prediction (CP) as a promising tool for such applications \cite{vovk2005algorithmic,lindemann2025formal}. Unlike methods that rely on strong distributional assumptions, CP provides rigorous finite-sample coverage guarantees. ACI extends CP to online, non-exchangeable settings by employing a control-theoretic update rule to adjust quantile thresholds, thereby maintaining long-run coverage under distribution shifts \cite{gibbs2021adaptive}. This has made ACI an attractive foundation for statistical risk control in dynamic environments. However, standard ACI suffers from a critical limitation we term severity-blindness. Because its update rule is driven by a binary miscoverage indicator, its response is identical for a near-miss and a catastrophic failure. The algorithm adapts to whether a violation occurred, not to its physical severity. In dynamic control systems, this is a significant risk: a series of severe violations can cause the system state to diverge faster than the algorithm can adapt. This limitation of standard ACI motivates the need for a refined framework.
We contend that runtime assurance necessitates simultaneous regulation of both statistical reliability (violation frequency) and physical harm (violation intensity).

In this paper, we introduce Cost-Aware ACI, a novel framework that integrates violation costs directly into the conformal adaptation process. Our key insight is that assurance margins should adapt not only to the rate of constraint violations but also to the operational severity of resulting constraint violations. By weighting prediction coverage failures according to their costs, our algorithm enables controllers to recover proportionally from dangerous situations, expanding prediction sets aggressively when needed while maintaining efficiency during normal operation. This creates a natural feedback mechanism where the system becomes more conservative precisely when—and only when—the physics demands it. Finally, we validate the efficacy of this framework through a series of experiments. Across all experiments, our method consistently achieves substantially lower long-run average violation costs and frequencies than standard ACI, providing a practical framework for managing risk in unknown environments. Moreover, the proposed method is robust to hyperparameter variations, underscoring its practical applicability. The contributions are summarized below:

$\bullet$ \textbf{Cost-Aware Conformal Adaptation:} We introduce a novel update rule based on a cost-aware loss function that couples miscoverage indicators with constraint-violation costs, enabling simultaneous control of both the frequency of constraint violations and their long-run average cumulative cost.

$\bullet$ \textbf{Practical Online Control Algorithm:} We develop a model-free algorithm for online risk-aware controller synthesis, termed \textit{Online Adaptive Conformal Cost Control}. This algorithm integrates our cost-aware update rule into a closed-loop synthesis process that effectively balances task performance with runtime risk control in unknown and dynamic environments.

$\bullet$ \textbf{Rigorous Theoretical Guarantees:} We provide rigorous theoretical proofs for the proposed framework. We derive explicit, non-asymptotic bounds on both the violation frequency and the average cumulative cost. Furthermore, we demonstrate the algorithm's adaptive conservatism and formalize the update rule as an online optimization procedure with sublinear regret.


\textbf{Related Work:} A wide array of control techniques falls under the umbrella of constrained control. Prominent model-based methods, including model predictive control (MPC) \cite{garcia1989model,morari1999model}, control barrier functions \cite{ames2016control,ames2019control}, funnel control \cite{bechlioulis2008robust}, and their temporal logic extensions \cite{lindemann2018control,farahani2018shrinking}, can provide strong guarantees, but usually require accurate system models and are hard to apply to complex or learning-enabled systems. Data-driven approaches like reinforcement learning \cite{ibarz2021train}, imitation learning \cite{torabi2018behavioral}, and learned controllers \cite{dean2021guaranteeing,taylor2020learning,srinivasan2020synthesis, robey2020learning,lindemann2024learning}, often perform well empirically but typically lack verifiable operational guarantees, creating a critical gap between flexibility and reliability.

CP \cite{vovk2005algorithmic,shafer2008tutorial} provides a path toward bridging this gap, which is a model-agnostic and distribution-free framework characterized by its rigorous finite-sample marginal coverage guarantee. It has recently been applied to control problems, for instance, in learning uncertainty-informed safety filters \cite{lindemann2023safe,strawn2023conformal} and predictive monitors \cite{cairoli2023conformal}, as surveyed in \cite{lindemann2025formal}. However, CP's validity hinges on the assumption of exchangeability, which is violated in non-stationary settings like time-series control. This has motivated adaptive methods for distribution shift, such as weighted conformal approaches \cite{tibshirani2019conformal,barber2023conformal}.

ACI \cite{gibbs2021adaptive} is a powerful control-theoretic solution for non-exchangeable time-series data \cite{zaffran2022adaptive,sousa2024general,szabadvary2024adaptive}. Drawing on online learning \cite{zinkevich2003online}, ACI adaptively tunes the prediction interval size to regulate the long-run average miscoverage rate to a user-specified level. Subsequent variants, including Conformal PID \cite{angelopoulos2023conformal}, Bellman Conformal Inference \cite{yang2024bellman}, and improved step-size adaptation \cite{gibbs2024conformal,podkopaev2024adaptive}, have improved its stability and responsiveness. While ACI has been integrated into critical domains such as motion planning \cite{dixit2023adaptive}, control barrier functions \cite{zhou2024safety}, safe online planning in POMDPs \cite{sheng2024safe}, and constrained reinforcement learning \cite{yao2024sonic}, its standard formulation only controls the frequency of violations, not magnitude. Besides, prior works have generalized ACI to control a user-defined average loss \cite{feldman2022achieving,lekeufack2024conformal}, but they typically treat this as a single-objective calibration problem: tuning a parameter to satisfy one aggregate bound. Our work addresses the distinct and more complex challenge of risk-aware control, where one must simultaneously regulate both the statistical reliability (violation frequency) and physical harm (violation intensity). We achieve this through a novel `Cost-Aware’ loss function, which explicitly couples the binary miscoverage error with the physical cost. 

\textbf{Notations:} $\mathbb{N}$ denotes the set of natural numbers;  $\mathbb{N}_{\geq i}$ denotes the set of natural numbers being larger than or equal to $i$; $\mathbb{R}$ denotes the set of real numbers; $\mathbb{R}^n$ denotes $n$-dimensional Euclidean space; $\mathbb{R}_{\geq 0}$ denotes the set of non-negative real numbers; $|\mathcal{A}|$ denotes the Cardinality of a set $\mathcal{A}$;  $\mathbb{I}(\cdot)$ denotes the indicator function, which evaluates to $1$ if the condition is true, and $0$ otherwise. 

\section{Preliminaries}
\label{sec:pre}
This section formalizes the problem of online risk-aware controller synthesis for systems with unknown, time-varying dynamics, and reviews ACI together with its limitations in this critical context.
\subsection{Problem Formulation}
\textbf{System Dynamics.} Consider a dynamical system described by the state update equation:
\begin{equation}
    \bm{x}_{k+1} = \bm{f}(\bm{x}_{k}, \bm{u}_k), 
    \label{eq:dynamics}
\end{equation}
where $k \in \mathbb{N}$ is the time step, $\bm{x}_k \in \mathcal{X} \subseteq \mathbb{R}^n$ is the system state, and $\bm{u}_k \in \mathcal{U} \subseteq \mathbb{R}^m$ is the control input. We operate under the following assumption.

\begin{assumption}[Unknown Dynamics]
\label{assmp:ud}
The transition function $\bm{f}: \mathcal{X} \times \mathcal{U} \to \mathbb{R}^n$ is deterministic but \textit{unknown} and may vary over time, reflecting a dynamic operational environment.
\end{assumption}

\textbf{Online Control Protocol.} The controller synthesis protocol operates under strict causal constraints. At each  step $k$, the control input $\bm{u}_k$ must be generated using only the current state $\bm{x}_k$ and historical observations $\{ (\bm{x}_i, \bm{u}_i) \}_{i=0}^{k-1}$. This online setting requires balancing between \textit{exploration} (to learn the system dynamics) and \textit{exploitation} (to achieve task objectives while maintaining assurance).

\textbf{Constraint and Cost Formulation.}
Constraint satisfaction is defined by a known scalar function $h: \mathcal{X} \to \mathbb{R}$, where the \textit{admissible set} (or nominal safe set) is its zero-superlevel set, $\mathcal{S} = \{ \bm{x} \in \mathcal{X} \mid h(\bm{x}) \ge 0 \}$. To quantify the severity of violations, we define a state-dependent cost function:
\begin{equation}
c(\bm{x}) =
\begin{cases}
    0, & \text{if } h(\bm{x}) \ge 0 \quad (\text{Within admissible set}), \\
    \rho(\bm{x}), & \text{if } h(\bm{x}) < 0 \quad (\text{Deviation from admissible set}),
\end{cases}
\label{eq:cost_function}
\end{equation}
where $\rho(\bm{x}) > 0$ is a known function that penalizes the magnitude of a constraint violation. We assume $\rho(\bm{x})$ is bounded. For theoretical and numerical stability, it is convenient to normalize this cost to the range $[0, 1]$ by scaling with its known upper bound, without loss of generality.

\textbf{Control Objectives.}
The overall objective is to optimize a task performance metric, $J_{\text{task}}$, while satisfying critical operational constraints. The metric $J_{\text{task}}$ can take various forms depending on the user-specified performance preferences. As strict, state-wise enforcement of the admissible set ($\bm{x}_k \in \mathcal{S}, \forall k$) is often infeasible in unknown and dynamic environments, we adopt a more pragmatic, two-part  risk definition that constrains both the long-run average violation cost and frequency:
\[
    \textstyle J = \limsup_{N \to \infty} \frac{1}{N} \sum_{k=1}^{N} c(\bm{x}_{k}) \leq D_{\text{tol}}, \quad
    V = \limsup_{N \to \infty} \frac{1}{N} \sum_{k=1}^{N} \mathbb{I}(c(\bm{x}_{k}) > 0) \leq \alpha, 
\]
where $D_{\text{tol}} \ge 0$ is a user-specified tolerance for the average cost and $\alpha \in (0, 1)$ is the maximum allowable violation frequency. This formulation transforms the rigid invariance requirement into two manageable ``risk budgets'', permitting sparse violations while ensuring their cumulative harm remains bounded.  Crucially, this ensures the control problem remains feasible even if the system briefly exits the admissible set, enabling the controller to minimize damage rather than failing. For many applications, this is the most practical and achievable type of runtime assurance.

\subsection{Adaptive Conformal Inference}

In this subsection, we briefly review ACI \cite{gibbs2021adaptive}, a control-theoretic extension of CP for non-exchangeable data.
Consider an online prediction setting with a sequence of covariate--target pairs $(\bm{x}_0,y_0), (\bm{x}_1,y_1), \ldots$, where $y_k \in \mathcal{Y}$. At each time step $k$, given the current covariate $\bm{x}_k$ and past observations, the goal is to construct a prediction set $\hat{C}_k\subseteq\mathcal{Y}$ for the unknown target $y_k$, targeting a long-run miscoverage rate $\alpha$.
Since exchangeability is typically violated in online settings, ACI does not fix the miscoverage level $\alpha$; instead, it maintains a time-varying parameter $\alpha_k$, updated by a gradient-like mechanism that decreases after coverage failures and increases otherwise.

At each step $k$, ACI uses a non-negative conformity score $s_k = S(\bm{x}_k, y_k; \hat{g}_{k-1})$ based on a predictive model $\hat{g}_{k-1}$. The empirical quantile function, $\hat{Q}_k(\cdot)$, is constructed from a calibration set of past scores, $\mathcal{D}_k^{\text{cal}} = \{s_i\}_{i < k}$:
\begin{equation*}
\label{eq:quantile_def}
\textstyle \hat{Q}_k(p) := \inf \big\{ s \in \mathbb{R}: \frac{1}{|\mathcal{D}_k^{\text{cal}}|} (\sum_{s_i \in \mathcal{D}_k^{\text{cal}}} \mathbb{I}(s_i \le s) ) \ge 1-p \big\},
\end{equation*}
where $p\in [0,1]$. ACI uses an adaptive internal variable, $\alpha_k$, to construct the prediction set $\hat{C}_k(1-\alpha_k)$ for a target variable $y_k$:
\begin{equation}
\label{eq:cover_comp}
\hat{C}_k(1-\alpha_k) := 
\begin{cases} 
    \mathcal{Y} & \text{if } \alpha_k \le 0, \\
    \emptyset & \text{if } \alpha_k \ge 1, \\
    \{ y \in \mathcal{Y} : S(\bm{x}_k, y; \hat{g}_{k-1}) \le \hat{Q}_k(\alpha_k) \} & \text{if } 0 < \alpha_k < 1.
\end{cases}
\end{equation}
The variable $\alpha_k$ is updated via a linear rule, starting with $\alpha_0 = \alpha$:
\begin{equation}
\label{eq:s_update}
\alpha_{k+1} := \alpha_k + \gamma (\alpha - e_k),
\end{equation}
where $e_k = \mathbb{I}(y_k \notin \hat{C}_k(1-\alpha_k))$ is the miscoverage indicator and $\gamma > 0$ is a step-size. This update ensures a powerful long-run average miscoverage guarantee.

\begin{proposition}[Long-run Average Miscoverage \cite{gibbs2021adaptive}]
\label{prop:alrc}
For the update rule in \eqref{eq:s_update}, the cumulative miscoverage rate converges to the target level $\alpha$, i.e., $\lim_{T\to\infty} \frac{1}{T} \sum_{k=0}^{T-1} e_k = \alpha$.
\end{proposition}

However, while ACI can provide a powerful tool to enforce the frequency constraint $V\leq \alpha$ as shown in Theorem \ref{thm:validity} below, its standard form is insufficient for our dual-objective risk control problem. The core issue is ``severity-blindness'': the binary nature of $e_k$ means the system's response is identical for a near-miss and a catastrophic failure. This decoupling from the physical cost $c(\bm{x}_k)$ means ACI makes no effort to satisfy the cost constraint $J\leq D_{\text{tol}}$. Although prior work generalized ACI to control a single user-defined loss \cite{feldman2022achieving}, our problem requires simultaneously regulating both constraint violation frequency and cost. This critical gap necessitates a new, cost-aware framework.

\section{Adaptive Conformal Cost Control}
\label{sec:method}
This section introduces Cost-Aware ACI, a framework designed to overcome the severity-blindness of standard ACI in runtime assurance. It explicitly integrates the cost of constraint violations into the adaptive mechanism, ensuring that uncertainty bounds adapt not only to a miscoverage event but also in proportion to its severity. We embed this framework into an online algorithm to synthesize a controller with dual statistical guarantees. It simultaneously bounds the long-run constraint violation frequency and average cumulative cost, facilitating rapid recovery from dangerous transients.

\subsection{Cost-Aware ACI}

We define a nonconformity score that quantifies prediction uncertainty in terms of constraint satisfaction. Let $\Phi: \mathcal{X} \times \mathcal{U} \rightarrow \mathbb{R}$ be the composite function that maps a state-action pair to the assurance value of the resulting state, i.e., $\Phi(\bm{x}, \bm{u}) := h(\bm{f}(\bm{x}, \bm{u}))$.
We learn a surrogate model $\widehat{h}(\bm{x}, \bm{u}) \approx \Phi(\bm{x}, \bm{u})$ from online data. The prediction error at time $k$, which results from applying control input $\bm{u}_{k-1}$ at the previous state $\bm{x}_{k-1}$, is the absolute residual:
\begin{equation}
    s_k := \left| \Phi(\bm{x}_{k-1}, \bm{u}_{k-1}) - \widehat{h}(\bm{x}_{k-1}, \bm{u}_{k-1}) \right| = \left| h(\bm{x}_{k}) - \widehat{h}(\bm{x}_{k-1}, \bm{u}_{k-1}) \right|,
    \label{eq:safety_score}
\end{equation}
where $k \in \mathbb{N}_{\geq 1}$. This residual, $s_k$, serves as our nonconformity score, measuring the discrepancy between predicted and realized runtime assurance. A large score indicates a significant prediction error that could compromise operational reliability. To guarantee the well-posedness of the online controllers synthesis optimization (\eqref{eq:robust_constraint} in Section \ref{sub:scs}), the prediction residuals must remain finite. Specifically, unbounded residuals could lead to infinitely large uncertainty margins, rendering the robust constraint unsatisfiable. This assumption is mild and practically justified for physical control systems. Since system states and control inputs are subject to inherent physical limits, both the true assurance function and its learned surrogate are naturally bounded, guaranteeing finite residuals. 
\begin{assumption}[Bounded Nonconformity]
\label{assumption0}
There exists a constant $M > 0$ such that $s_k \le M$ for all $k \in \mathbb{N}$.
\end{assumption}


Building on the nonconformity score, we define a loss function that incorporates the cost of a constraint violation. First, we define the empirical $(1-\delta)$-quantile function, $\hat{Q}_k$, of the historical scores $\mathcal{C} = \{s_1, \dots, s_{k-1}\}$ and a value $\epsilon > 0$:
\begin{equation*}
    \hat{Q}_k(\delta) := 
\begin{cases}
M, & \text{if~}\delta < 0,\\
-\epsilon, & \text{if~}\delta > 1,\\
\inf \left\{ c \in \mathbb{R} : \frac{1}{|\mathcal{C}|} \sum_{r \in \mathcal{C}} \mathbb{I}(r \le c) \ge 1-\delta \right\}, & \text{if~}0\leq\delta\leq 1.
\end{cases}
\end{equation*}
Standard ACI uses a binary miscoverage indicator, $e_k = \mathbb{I}(s_k > \hat{Q}_k(\delta_k))$, as feedback. To introduce severity awareness, we augment this indicator with the constraint violation cost, $c(\bm{x}_k)$, defining the \textbf{boosted loss}:
\begin{equation} \label{eq:boosted_loss}
    L_k := e_k \cdot \left( 1 + \beta c(\bm{x}_{k}) \right).
\end{equation}
Equivalently, since $c(\bm{x}_{k})= 0$ when $h(\bm{x}_{k}) \ge  0$ and $c(\bm{x}_{k}) = \rho(\bm{x}_{k})>0$ when $h(\bm{x}_{k})<0$, we have
\begin{equation*}
    L_k = 
    \begin{cases} 
        0 & \text{if } s_k \le \hat{Q}_{k}(\delta_k), \\
        1 & \text{if } s_k > \hat{Q}_{k}(\delta_k) \text{ and } h(\bm{x}_{k}) \ge 0, \\
        1 + \beta \rho(\bm{x}_{k}) & \text{if } s_k > \hat{Q}_{k}(\delta_k) \text{ and } h(\bm{x}_{k}) < 0.
    \end{cases}
\end{equation*}
The parameter $\beta \ge 0$ is a severity-sensitivity parameter that controls the trade-off between frequency and severity of constraint violations:

$\bullet ~ \beta = 0$: Recovers standard ACI, where all failures are weighted equally regardless of severity.

$\bullet ~ \beta > 0$: Introduces cost awareness. A miscoverage event that causes constraint violations ($h(\bm{x}_k)<0$) is penalized more heavily, in proportion to $\rho(\bm{x}_k)$.


With the boosted loss defined, we introduce the update mechanism for the adaptive miscoverage level $\delta_{k+1}$. This parameter is adjusted via online gradient descent on the loss $L_k$:
\begin{equation} \label{eq:update_rule}
\delta_{k+1} = \delta_{k} + \gamma (\alpha - L_k),
\end{equation}
where $\alpha \in (0,1)$ is the target risk budget and $\gamma > 0$ is the learning rate. This update rule creates three distinct operational regimes based on the outcome at time step $k$:

\textbf{1. Within Admissible Set ($L_k = 0$)}: When the prediction is accurate ($s_k \le \hat{Q}_k$), the loss is zero. The change in $\delta$, denoted $\Delta\delta_{k+1} = \delta_{k+1} - \delta_k$, is positive: $\Delta\delta_{k+1} = \gamma\alpha$. This gradually increases $\delta_k$, tightening the uncertainty bounds to improve control performance.

\textbf{2. Benign Failure ($L_k = 1$)}: When a miscoverage event occurs but the system remains within the admissible set ($h(\bm{x}_k) \ge 0$), the loss is 1. The change is negative: $\Delta\delta_{k+1} = -\gamma(1-\alpha)$. This penalizes the prediction error by making the uncertainty bounds more conservative.

\textbf{3. Critical Failure ($L_k > 1$)}: When miscoverage coincides with a constraint violation ($h(\bm{x}_k) < 0$), the loss is $1 + \beta \rho(\bm{x}_k)$. The change becomes sharply negative: $\Delta\delta_{k+1} = -\gamma(1 - \alpha + \beta \rho(\bm{x}_k))$. The decrease is now proportional to the violation's severity, triggering a strong corrective response.

The key innovation is this cost-sensitive adaptation during critical failures. While standard ACI uses a fixed correction step, Cost-Aware ACI employs a cost-scaled update. This enables proportional recovery: severe violations trigger stronger corrections, reducing the risk of prolonged violations.


\subsection{Run-Time Assurance Control Synthesis}
\label{sub:scs}

Having established the adaptive uncertainty quantification, we now integrate it into online control synthesis. At each step $k$, the newly calibrated parameters ensure the next action $\bm{u}_k$ is admissible. The calibrated quantile $\hat{Q}_{k+1}(\delta_{k+1})$ provides a robust lower bound on the true assurance value:
\begin{equation*}
\Phi(\bm{x}_k, \bm{u}_k) \geq \widehat{h}(\bm{x}_k, \bm{u}_k) - \hat{Q}_{k+1}(\delta_{k+1}).
\end{equation*}
To provide run-time assurance, we require the predicted assurance value to be non-negative, leading to the robust constraint \(\widehat{h}(\bm{x}_k, \bm{u}_k) \geq \hat{Q}_{k+1}(\delta_{k+1}).\) Consequently, at each time step $k$, we determine the optimal control input $\bm{u}^*_k$ by solving the following MPC optimization problem:
\begin{equation}
\begin{split}
(\bm{u}^*_k, \{\bm{u}^{**}_i\}_{i=k+1}^{k+N-1})&=\argmin_{\{\bm{u}_i\}_{i=k}^{k+N-1}} J_{\text{task}}(\{\bm{x}_i\}_{i=k+1}^{k+N},\{\bm{u}_i\}_{i=k}^{k+N-1})\\
s.t. &\begin{cases}
\widehat{h}(\bm{x}_k, \bm{u}_k) \geq \hat{Q}_{k+1},\\
\bm{x}_{i+1}=\widehat{\bm{f}}(\bm{x}_i,\bm{u}_i),i=k,\ldots,k+N-1.
\end{cases}
\end{split}
\label{eq:robust_constraint}
\end{equation}
Here, $\widehat{\bm{f}}: \mathcal{X}\times \mathcal{U}\rightarrow \mathbb{R}^n$ serves as a surrogate model for the transition function $\bm{f}: \mathcal{X}\times \mathcal{U}\rightarrow \mathbb{R}^n$. It is worth noting that we employ a separate estimator $\widehat{h}$ to approximate $\Phi$, rather than relying on the composition $h(\widehat{\bm{f}}(\cdot))$. This design decouples assurance prediction errors from state dynamics errors. By directly calibrating the scalar residual of $\widehat{h}$, we would avoid the excessive conservatism often caused by propagating high-dimensional state uncertainty through nonlinear functions. The  constraint $\widehat{h}(\bm{x}_k, \bm{u}_k) \geq \hat{Q}_{k+1}$ enforces a reliability level that is dynamically adjusted by the uncertainty margin $\hat{Q}_{k+1}$. Meanwhile, state evolution within the planning horizon is projected using the surrogate dynamics model $\widehat{\bm{f}}: \mathcal{X}\times \mathcal{U}\rightarrow \mathbb{R}^n$. Both surrogate models, $\widehat{h}$ and $\widehat{\bm{f}}$, can be continuously updated online to accurately reflect the current environment.

This optimization establishes a critical feedback loop that balances task performance, prediction reliability, and control conservatism. By adaptively tightening the  margin $\hat{Q}_{k+1}$ only when necessary, the system secures reliability during critical environmental shifts while minimizing conservatism during stable operation to maximize task performance.


The complete online adaptive conformal cost control algorithm integrates these components into a cohesive online control framework, outlined in Algorithm \ref{alg:risk_sensitive_ACI}, with further details in Appendix~\ref{sec:app_alg}.





\begin{algorithm}[!h]
\caption{Online Adaptive Conformal Cost Control Algorithm}
\label{alg:risk_sensitive_ACI}
\begin{algorithmic}[1]
   \REQUIRE Target risk budget $\alpha$, learning rate $\gamma$, sensitivity $\beta$;
   \STATE Initialize $\delta_1 = \alpha$, $\mathcal{C} = \{\}$, models $\widehat{h}$ and $\widehat{\bm{f}}$; Observe initial state $\bm{x}_0$;
   \FOR{$k=0,1, \ldots$}
       \IF{$k > 0$}
           \STATE \textbf{Scoring:} $s_k = |h(\bm{x}_k) - \widehat{h}(\bm{x}_{k-1}, \bm{u}_{k-1}^*)|$;
           \STATE \textbf{Severity Evaluation:} $e_k = \mathbb{I}(s_k > \hat{Q}_{k}(\delta_k))$;
           \STATE \hspace{3.08cm} $c(\bm{x}_k) = \rho(\bm{x}_k) \cdot \mathbb{I}(h(\bm{x}_k) < 0)$;
           \STATE \hspace{3.08cm} $L_k = e_k(1 + \beta c(\bm{x}_k))$;
           \STATE \textbf{Update:} $\delta_{k+1} = \delta_{k} + \gamma(\alpha - L_k)$;
           \STATE \hspace{1.27cm} $\mathcal{C} \leftarrow \mathcal{C} \cup \{s_k\}$;
           \STATE \hspace{1.27cm} Update $\widehat{h}$ and $\widehat{\bm{f}}$ (online learning);
       \ENDIF
       \STATE \textbf{Calibration:} $\hat{Q}_{k+1} \leftarrow \text{Quantile}(\{-\epsilon,M\}, 1-\delta_{k+1})$ if $k=0$, else $\text{Quantile}(\mathcal{C}, 1-\delta_{k+1})$;
       \STATE \textbf{Control Synthesis:} Solve optimization \eqref{eq:robust_constraint} using $\hat{Q}_{k+1}$ to obtain $\bm{u}_k^*$;
       \STATE Apply $\bm{u}_k^*$ and observe $\bm{x}_{k+1}$;
   \ENDFOR
\end{algorithmic}
\end{algorithm}

\section{Theoretical Analysis}
\label{sec:analysis}

This section establishes the theoretical properties of the proposed algorithm. We first introduce mild structural assumptions, then prove the stability of the algorithm by showing that its adaptive parameter remains bounded (Lemma \ref{lemma:boundedness}). Third, we establish its core convergence property: the long-run average boosted loss converges to the target level $\alpha$ (Lemma \ref{lemma:convergence}). Building on this convergence, we derive our main performance guarantees: rigorous upper bounds on both long-run constraint violation frequency (Theorem \ref{thm:validity}) and average cumulative constraint violation cost (Theorem \ref{thm:cost-bound}). Finally, we analyze the algorithm's adaptive nature, showing it dynamically adjusts to environmental hostility (Proposition~\ref{thm:conservation}) and corresponds to a principled online optimization procedure (Proposition~\ref{thm:regret}). All proofs are shown in Appendix A. We begin by stating the assumptions required for our analysis.

\begin{assumption}[Bounded Constraint Violation Costs]
\label{assumption:cost}
The constraint violation cost is normalized and bounded, i.e., $0 \leq \rho(\bm{x}) \leq 1$ for all $\bm{x} \in \mathcal{X}$. Consequently, the boosted loss is bounded by $L_{\max} = 1 + \beta$.
\end{assumption}

\begin{assumption}[Feasibility of the Robust Constraint]
\label{assumption:feasibility}
For every state $\bm{x}$ reachable by the system, optimization \eqref{eq:robust_constraint} is feasible. That is, the set of control inputs satisfying the constraint $\widehat{h}(\bm{x}, \bm{u}) \geq \hat{Q}_{k+1}$ is non-empty. 
\end{assumption}

Assumption \ref{assumption:cost} is a practical condition reflecting that physical damage is typically finite. Assumption \ref{assumption:feasibility} is a standard condition ensuring the control synthesis is well-defined, similar to its use in MPC \cite{scokaert2002suboptimal}. In practice, the feasibility required by Assumption \ref{assumption:feasibility} is structurally achieved by the availability of a robust recovery policy (e.g., emergency braking or hovering). This hierarchical architecture, which prioritizes a high-performance optimizer while retaining a conservative backup, parallels standard shielding and recovery mechanisms in Safe Reinforcement Learning \cite{alshiekh2018safe,thananjeyan2021recovery}. Since this recovery policy is designed to satisfy constraints under worst-case uncertainty, its presence ensures that the feasible set for the optimization \eqref{eq:robust_constraint} is never empty. While this guarantees solvability, our Cost-Aware algorithm is designed to proactively adjust margins to prioritize high-performance actions, thereby minimizing the frequency with which the controller is forced to select this conservative recovery option. With these assumptions, we first prove that the adaptive parameter $\delta_k$ remains stable.

\begin{lemma}[Parameter Boundedness]
\label{lemma:boundedness}
Let $\{\delta_k\}$ be generated by the update rule \eqref{eq:update_rule}. Under Assumptions \ref{assumption:cost} and \ref{assumption:feasibility}, if $\delta_1$ is initialized within the interval $[-\gamma(L_{\max} - \alpha), 1 + \gamma\alpha]$, then $\delta_k$ remains within this interval for all $k \in \mathbb{N}$.
\end{lemma}



This stability is crucial, as it prevents erratic behavior and is a prerequisite for proving convergence. We now prove the fundamental convergence property of Cost-Aware ACI: the long-run average cumulative boosted loss converges exactly to the target risk budget $\alpha$.

\begin{lemma}[Long-Run Average Boosted Loss Convergence]
\label{lemma:convergence}
Under the same conditions as Lemma \ref{lemma:boundedness}, the long-run time-average of the boosted loss converges to the target level $\alpha$, i.e., $\lim_{T \to \infty} \frac{1}{T} \sum_{k=1}^{T} L_k = \alpha$.
\end{lemma}


This convergence of the boosted loss is the key to unlocking our main theoretical results: explicit upper bounds on the long-run constraint violation frequency $V$ and average cost $J$.

\begin{theorem}[Upper Bound on Constraint Violation Frequency]
\label{thm:validity}
Let $v_k:=\mathbb{I}(h(\bm{x}_k) < 0)$ be the indicator of a constraint violation at time $k$. Under Assumptions \ref{assumption0} and \ref{assumption:cost}, the  violation frequency $V$ is upper-bounded by $\alpha$, i.e., $V = \limsup_{T \to \infty} \frac{1}{T} \sum_{k=1}^{T} v_k \le \alpha$.
\end{theorem}
Theorem \ref{thm:validity} shows that our cost-aware algorithm bounds the long-run constraint violation frequency $V$. Next, we establish a bound on the cumulative cost $J$.
\begin{theorem}[Upper Bounds on Average Cumulative Cost $J$]
\label{thm:cost-bound}
Under Assumptions \ref{assumption0} and \ref{assumption:cost}, for any risk sensitivity parameter $\beta > 0$, long-run average cumulative constraint violation cost $J$ satisfies $J= \limsup_{T \to \infty} \frac{1}{T} \sum_{k=1}^{T} c(\bm{x}_k) \le \frac{\alpha}{\beta}$. 
\end{theorem}
Theorems \ref{thm:validity} and \ref{thm:cost-bound} provide a direct method for online control design with run-time assurance. A designer can first set the maximum tolerable constraint violation frequency by choosing $\alpha$ (since $V \le \alpha$). Subsequently, to ensure the average cost $J$ remains below a specified tolerance $D_{\text{tol}}$, the sensitivity parameter $\beta$ is chosen to satisfy the condition $\alpha/\beta \le D_{\text{tol}}$.

The preceding theorems provide long-run guarantees. Such finite-time upper bounds for both the constraint violation frequency and the average cumulative cost can be directly obtained by leveraging the results of Lemma \ref{lemma:boundedness}.

\begin{corollary}[Finite-Time Violation Frequency and Cost Bounds]
\label{coro:violation}
Let $V_T:=\frac{1}{T}\sum_{k=1}^T v_k$ be the violation frequency and $J_T:=\frac{1}{T}\sum_{k=1}^T c(\bm{x}_k)$ be the average cumulative cost over the first $T$ time steps. The following bounds hold:

\textbf{1. Violation Frequency:} $V_T \le \alpha + \frac{\delta_1 - \delta_{T+1}}{T\gamma} = \alpha + O(1/T)$.

\textbf{2. Average Cost:} $J_T \le \frac{\alpha}{\beta} + \frac{\delta_1 - \delta_{T+1}}{\beta T\gamma} - \frac{1}{\beta T} \sum_{k=1}^T e_k \le \frac{\alpha}{\beta} + O(1/T)$.
\end{corollary}

This corollary provides a stronger, more practical guarantee. It shows that the violation frequency and average cost are tightly concentrated around their long-run targets, $\alpha$ and $\alpha/\beta$, respectively. The deviation from these targets is controlled by a transient term that diminishes at a rate of $O(1/T)$, ensuring that the system's performance rapidly converges to its desired steady-state behavior.

The next propositions offer deeper insight into the algorithm's adaptive behavior, framing it first as a rational economic agent and second as a principled online learning procedure.

\begin{proposition}[Conservation of Risk Budget]
\label{thm:conservation}
The algorithm dynamically trades off miscoverage frequency against violation severity. Let $\hat{\alpha}_T:=\frac{1}{T}\sum_{k=1}^T e_k$ be the empirical miscoverage rate and $\bar{c}_T:= \frac{\sum_{k=1}^T e_k c(\bm{x}_k)}{\sum_{k=1}^T e_k}$ be the average cost conditional on a miscoverage event. In the long run, these quantities adhere to the conservation law:
$\lim_{T \to \infty} \hat{\alpha}_T (1 + \beta \bar{c}_T) = \alpha$.
\end{proposition}

This proposition reveals the algorithm's economic rationality. It treats $\alpha$ not as a simple frequency target, but as a fixed risk budget. The convergence condition can be rewritten as $\hat{\alpha}_T \longrightarrow \frac{\alpha}{1 + \beta \bar{c}_T}$, showing how this budget is spent.
This expression demonstrates a dynamic trade-off. In benign environments where the average violation cost $\bar{c}_T$ is low, the denominator approaches 1, and the miscoverage rate $\hat{\alpha}_T$ converges toward the full budget $\alpha$. However, in hostile environments where $\bar{c}_T$ is high, the term $(1 + \beta\bar{c}_T)$ acts as an inflation factor. The ``price'' of each miscoverage event increases, compelling the algorithm to automatically reduce its miscoverage frequency $\hat{\alpha}_T$ to stay within its fixed budget $\alpha$. This mechanism enforces harm regulation by shifting from a purely statistical guarantee (bounded failure rate) to an operational one (bounded average cost).



Finally, we establish that this adaptive behavior is not heuristic but is grounded in a principled online optimization framework.

\begin{proposition}[Regret Minimization Perspective]
\label{thm:regret}
The update rule \eqref{eq:update_rule} is an instance of Online Gradient Descent (OGD) on the sequence of linearized loss functions $\mathcal{L}_k(\delta) = (L_k - \alpha)\delta$. Consequently, the algorithm achieves sublinear regret against the best fixed parameter $\delta^\star$ in hindsight: $\text{Regret}_T = \sum_{k=1}^{T} \mathcal{L}_k(\delta_{k}) - \min_{\delta^\star \in [0,1]} \sum_{k=1}^{T} \mathcal{L}_k(\delta^\star) \le O(\sqrt{T})$.
\end{proposition}

This regret-minimization perspective provides a solid theoretical foundation, showing that the algorithm is efficient and its performance approaches that of the best fixed parameter chosen in hindsight. Moreover, as noted in \cite{gibbs2024conformal}, regret guarantees rule out trivial strategies that alternate between overly permissive and overly conservative predictions, ensuring meaningful and stable adaptation.

\section{Experiments}
\label{sec:experiments}
In this section, we validate the effectiveness of the proposed Online Adaptive Conformal Cost Control Algorithm through a series of numerical experiments on four challenging control benchmarks: the Vanderpol oscillator \cite{henrion2013convex}, Pendulum~\cite{towers2024gymnasium}, MountainCar~\cite{towers2024gymnasium}, and Lorenz~\cite{lorenz1996predictability}. Our primary objective is to demonstrate that our method provides superior runtime assurance in terms of both violation frequency and cumulative cost when compared to standard ACI. In Appendix~\ref{sec:app_sensitivity}, we present a \textbf{sensitivity analysis} of the proposed algorithm with respect to $\gamma$, $\beta$, and the sliding-window size.

\noindent\textbf{Comparison Methods:} We evaluate the proposed Cost-Aware ACI in Alg.~\ref{alg:risk_sensitive_ACI} against the following three methods: 1. \textbf{\textit{Standard ACI}}: A variant of our algorithm where the cost-aware mechanism is disabled ($\beta=0$), representing the conventional ACI approach. 2. \textbf{\textit{Conformal PID}}: A variant of our method in which the ACI scheme is replaced by the conformal PID method of~\cite{angelopoulos2023conformal}. 3. \textbf{\textit{No ACI}}: A baseline where the ACI module is removed entirely by fixing the threshold $\hat{Q}_k \equiv 0$.

To create a rigorous evaluation, the task performance objective $J_{\text{task}}$ in each benchmark is intentionally designed to be antagonistic to constraint satisfaction, meaning that achieving better task performance incentivizes the system to leave the designated admissible set. This creates a realistic and challenging trade-off between performance and runtime assurance, stress-testing each algorithm's ability to balance competing objectives. Each simulation is run for $10,000$ time steps. Further implementation details are provided in Appendix~C. Fig.~\ref{fig:results} compares the four methods across 10 random seeds with $\alpha = 0.1$, whereas Table~\ref{tab:results} reports the results on MountainCar under different choices of $\alpha$. Our analysis focuses on two key comparisons:

\begin{figure}[h!]
    \centering
    \subfigure[Standard ACI]{\includegraphics[width=0.45\linewidth]{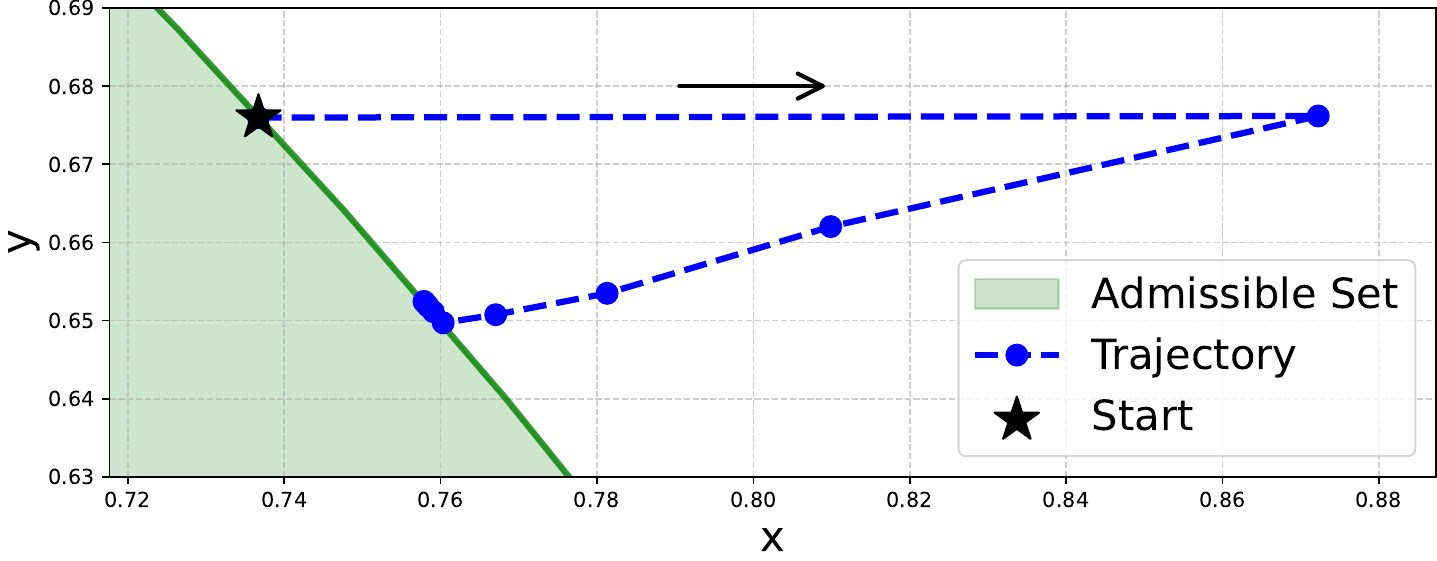}
    \label{fig:aci}
    }
    \subfigure[Ours]{\includegraphics[width=0.45\linewidth]{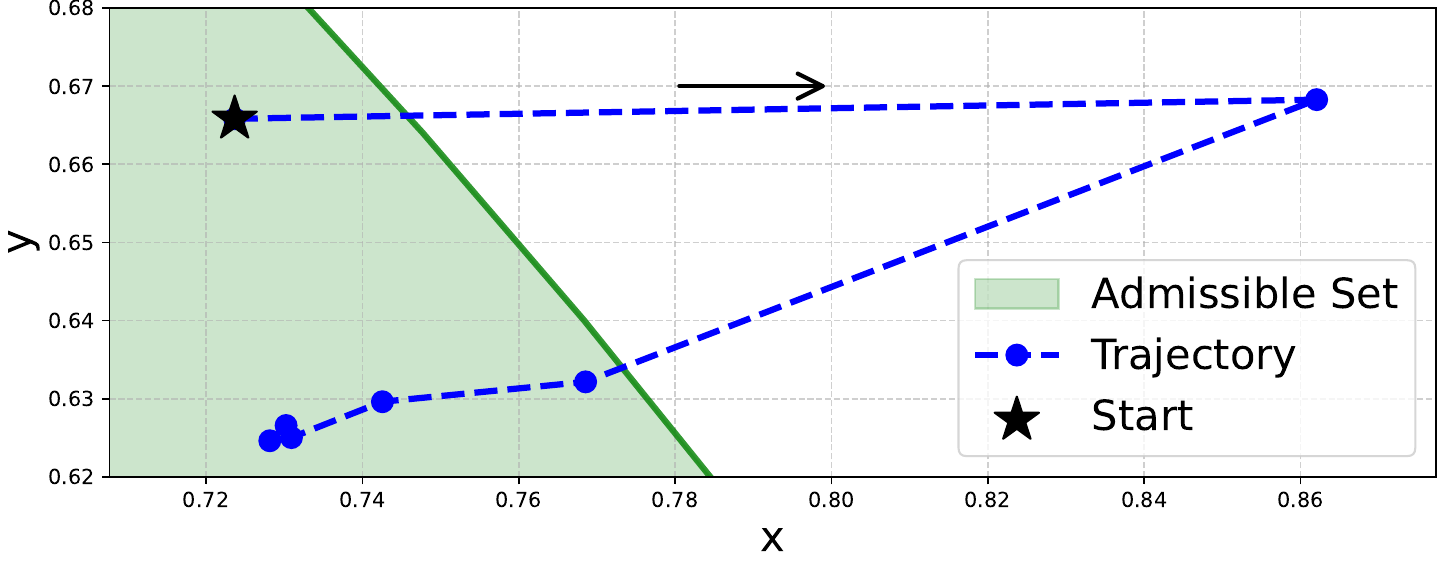}
    \label{fig:ours}
    }
    \caption{Comparison of selected trajectories on Vanderpol}
    \label{fig:traj}
\end{figure}


\noindent\textbf{The Benefit of ACI}: Comparing all ACI-based methods to the baseline without ACI shows the value of adaptive assurance constraints. The ``No ACI'' baseline
suffers from catastrophic failure rates across all benchmarks. In contrast, all ACI methods successfully regulate this risk. This confirms that Alg.~\ref{alg:risk_sensitive_ACI} provides a robust foundation for enhancing runtime assurance.

\noindent \textbf{Superiority of Cost-Aware ACI}: The central claim of our work is that Cost-Aware ACI provides superior runtime assurance. The results in Fig.~\ref{fig:results} and Table \ref{tab:results} strongly support this. Across all benchmarks and all values of $\alpha$, our method consistently and significantly outperforms standard ACI and Conformal PID on both metrics ($J$ and $V$). This performance gap is a direct result of our cost-sensitive update rule, which penalizes severe violations more aggressively, compelling the controller to adopt a more cautious policy. From Table~\ref{tab:results}, as predicted by our theory, this advantage becomes more pronounced for larger $\alpha$, where the algorithm has a larger ``risk budget'' that the standard ACI is more likely to misuse on severe violations.

The improved runtime assurance comes at a slight cost to task performance ($\bar{J}_{task}$). This trade-off is expected and even desirable, as our benchmarks were designed to make constraint satisfaction and performance compete. The fact that our method achieves substantial reductions in violation frequency and cost for only a minor performance loss highlights its efficiency. 


\begin{figure}[h!]
    \vspace{-10pt}
    \centering
    \subfigure[Vanderpol: Task Performance]{\includegraphics[width=0.30\linewidth]{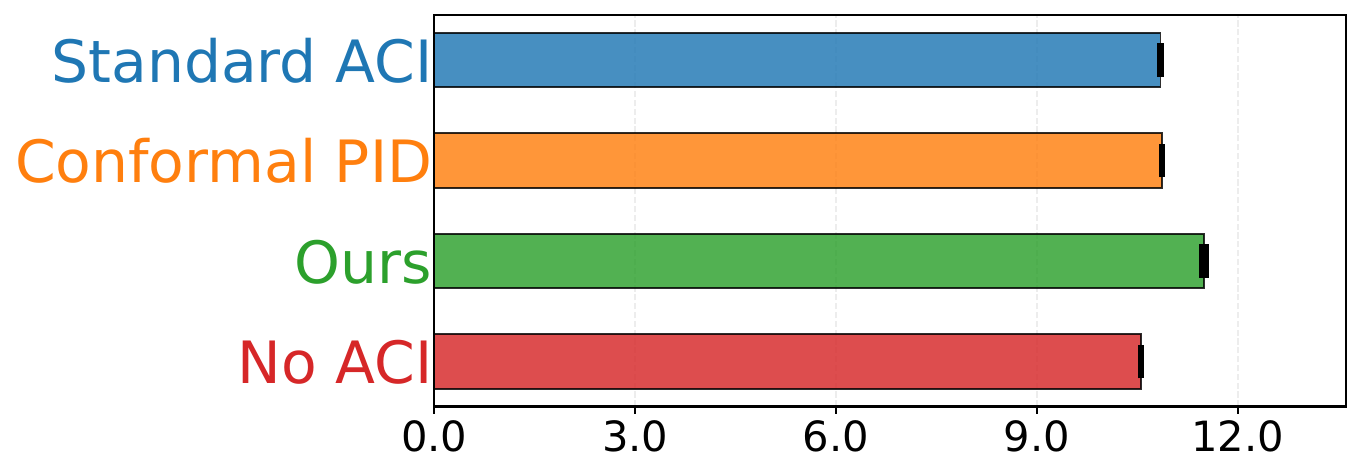}
    }
    \subfigure[Vanderpol: Violation Cost]{\includegraphics[width=0.30\linewidth]{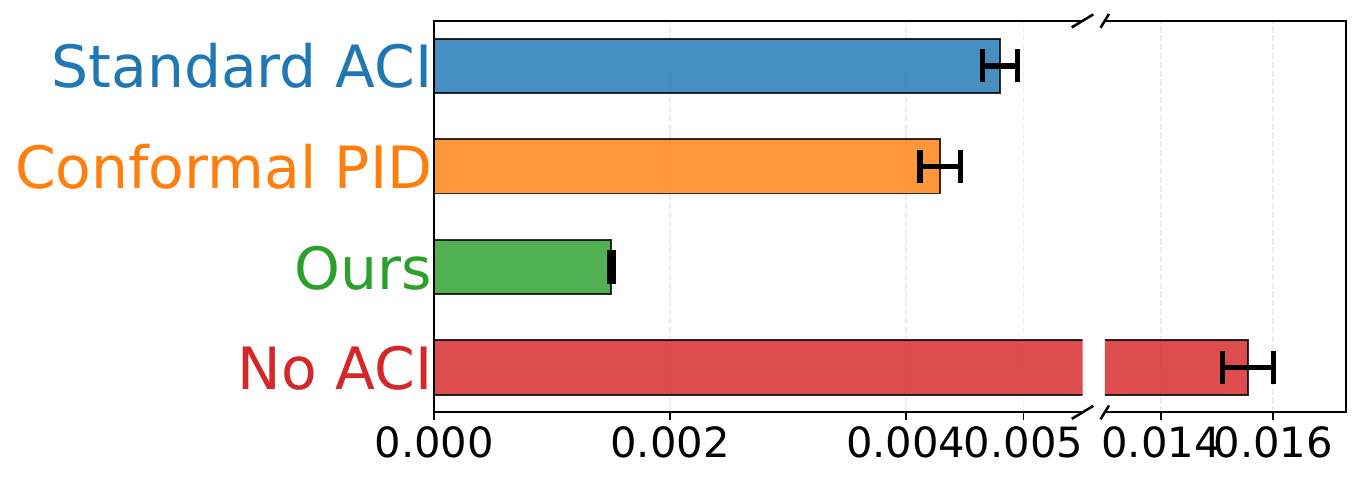}
    }
    \subfigure[Vanderpol: Violation Frequency]{\includegraphics[width=0.30\linewidth]{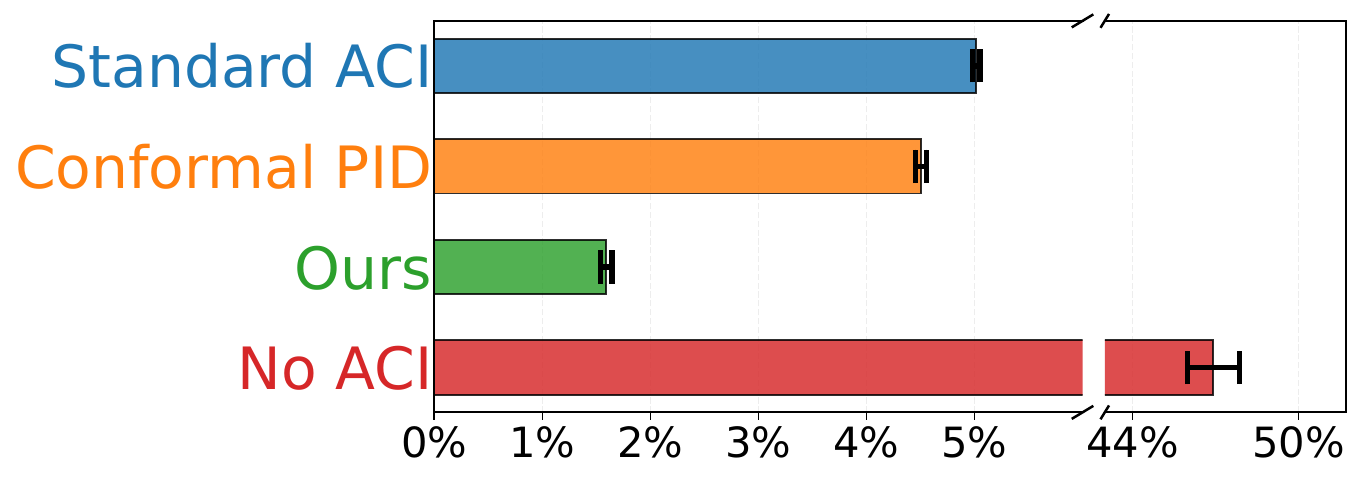}
    }\\
    \subfigure[Pendulum: Task Performance]{\includegraphics[width=0.30\linewidth]{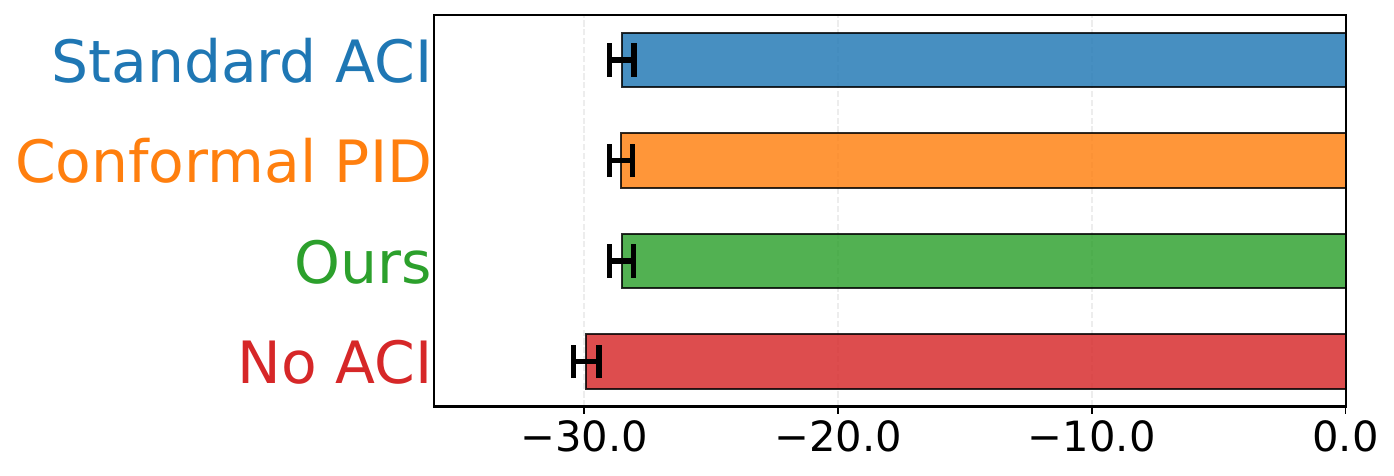}
    }
    \subfigure[Pendulum: Violation Cost]{\includegraphics[width=0.30\linewidth]{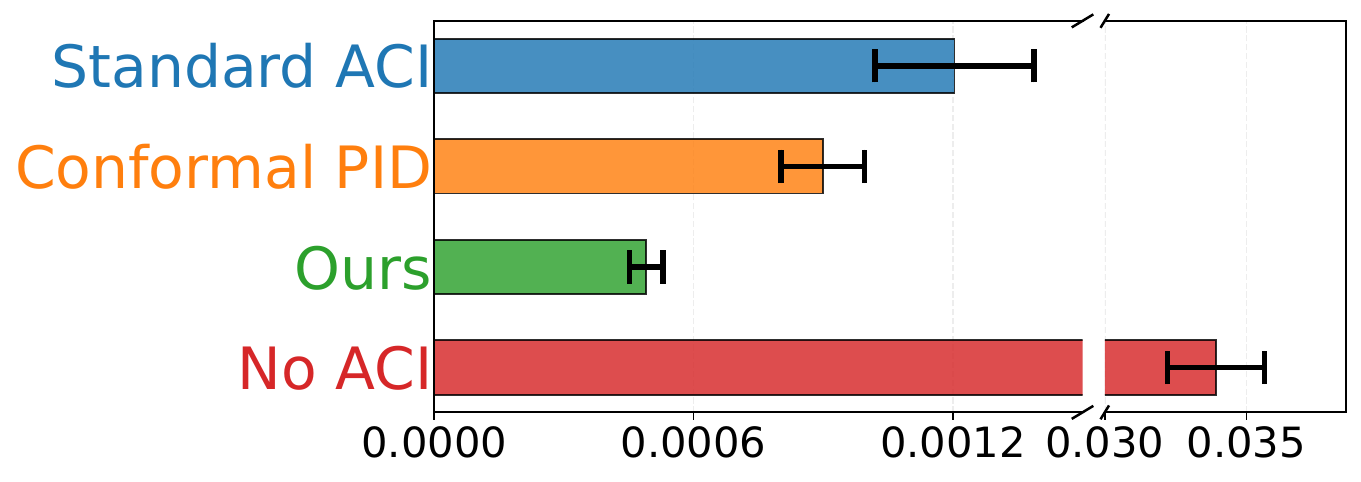}
    }
    \subfigure[Pendulum: Violation Frequency]{\includegraphics[width=0.30\linewidth]{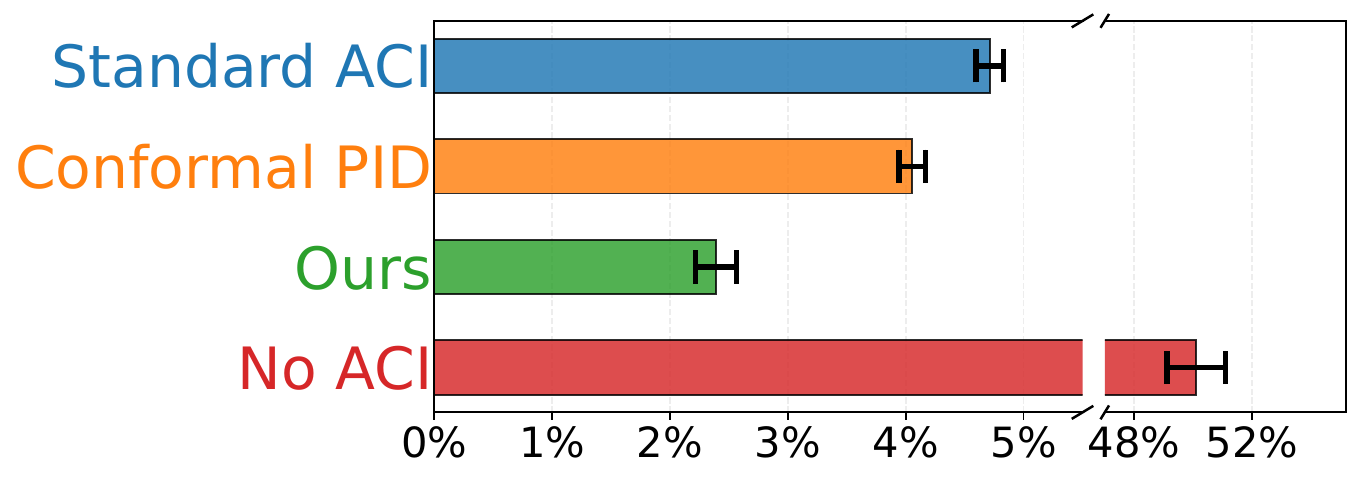}
    }\\
    \subfigure[MountainCar: Task Performance]{\includegraphics[width=0.30\linewidth]{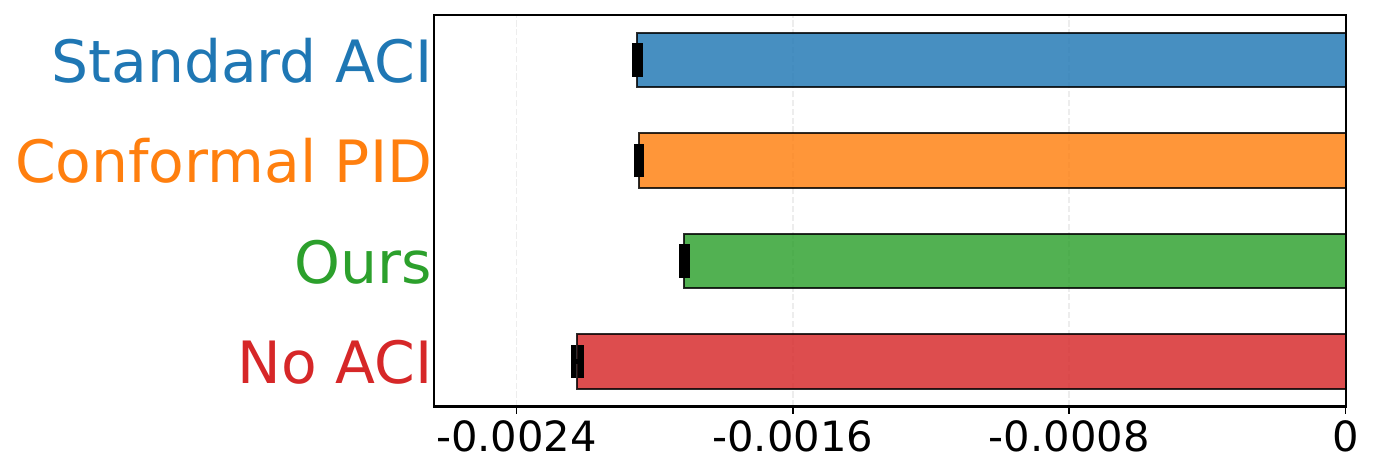}
    }
    \subfigure[MountainCar: Violation Cost]{\includegraphics[width=0.30\linewidth]{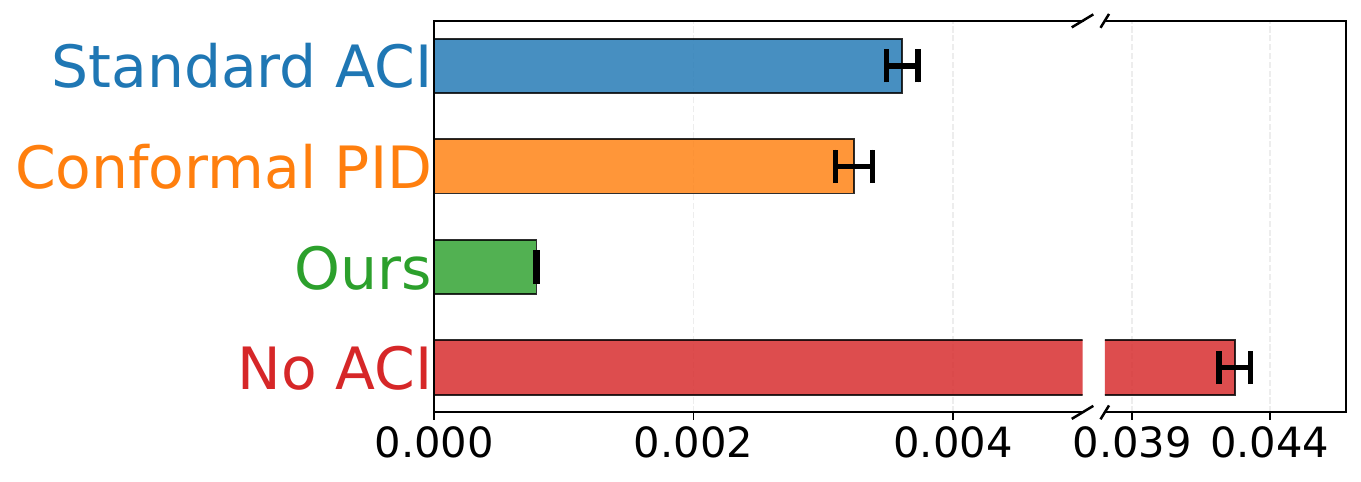}
    }
    \subfigure[MountainCar:Violation Frequency]{\includegraphics[width=0.30\linewidth]{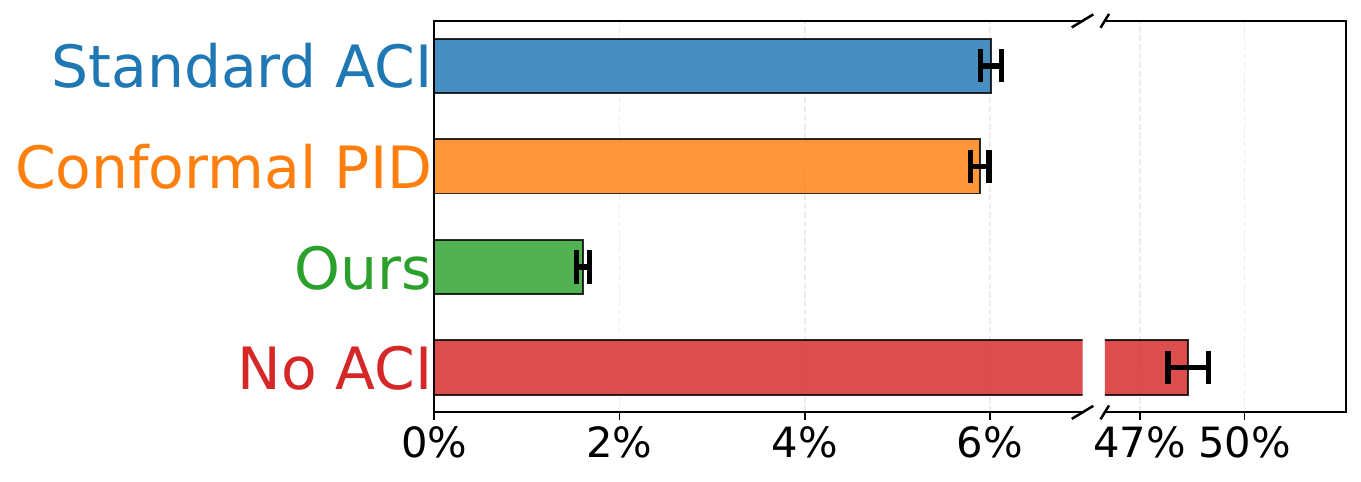}
    }\\
    \subfigure[Lorenz: Task Performance]{\includegraphics[width=0.30\linewidth]{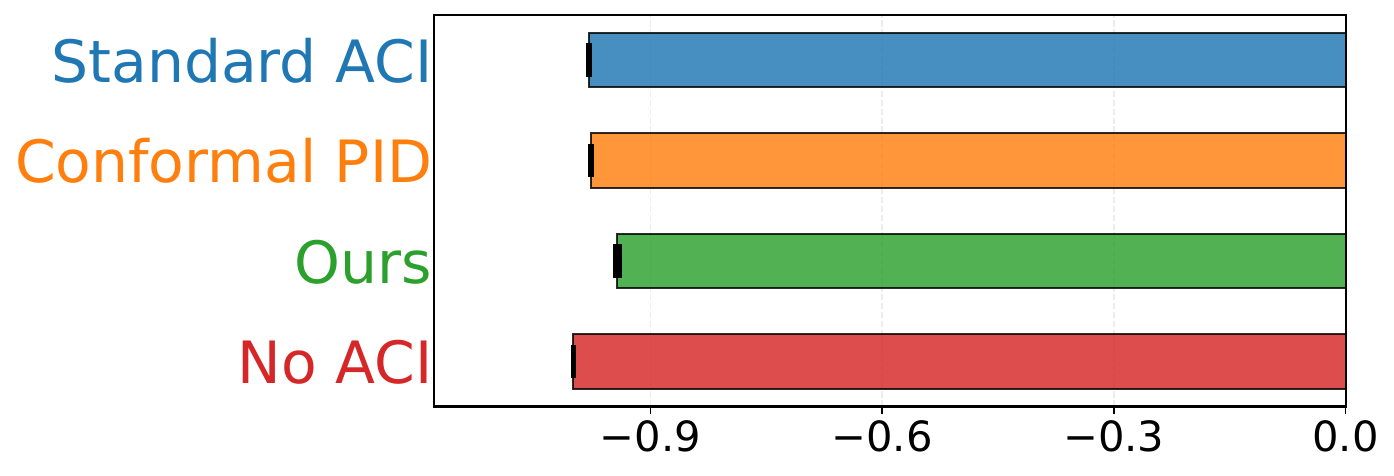}
    }
    \subfigure[Lorenz: Violation Cost]{\includegraphics[width=0.30\linewidth]{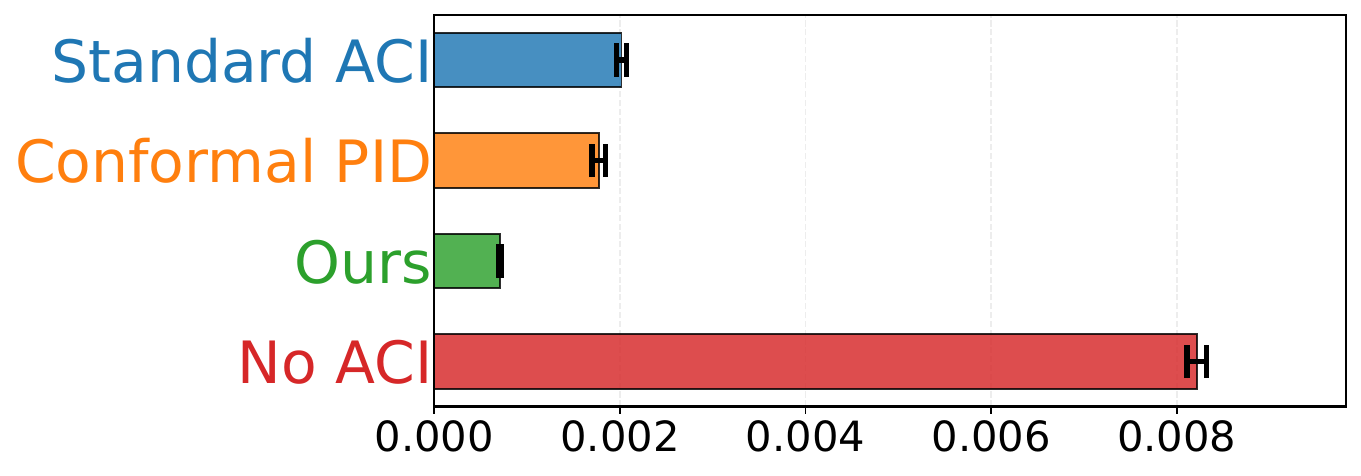}
    }
    \subfigure[Lorenz: Violation Frequency]{\includegraphics[width=0.30\linewidth]{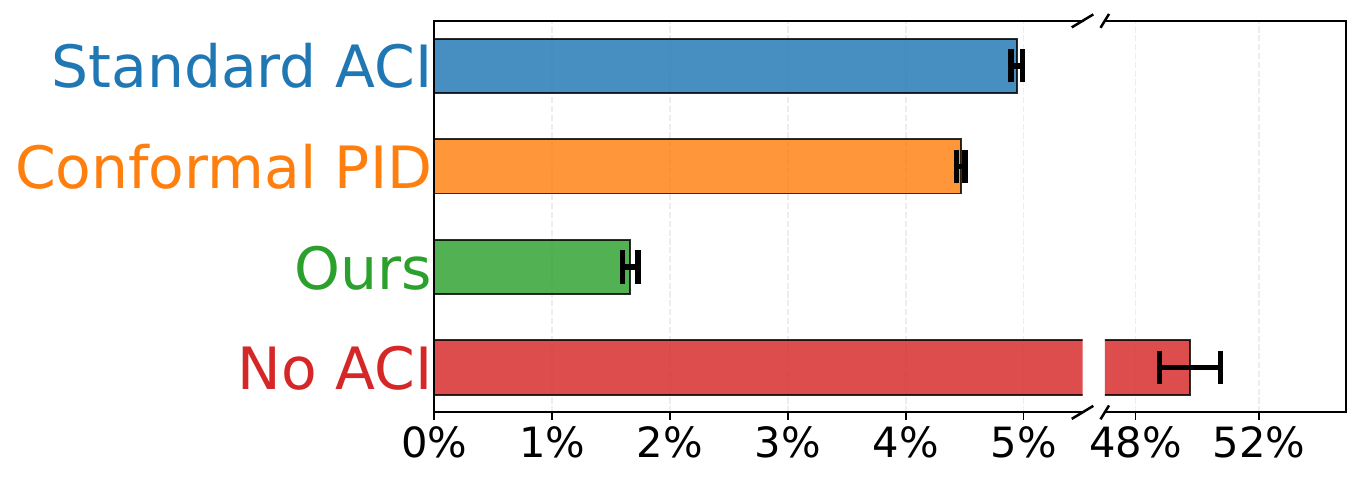}
    }
    \caption{Comparison of experimental results (\textbf{lower is better for all three metrics})}
    \label{fig:results}
\end{figure}
\begin{table*}[!h]
\centering
\caption{Performance comparison under different $\alpha$ on the MountainCar benchmark}
\footnotesize
\setlength{\tabcolsep}{0.9mm}{
    \begin{tabular}{*{10}{c}}
    \toprule
    \multirow{2}{*}{$\alpha$} & \multicolumn{3}{c}{Standard ACI} & \multicolumn{3}{c}{Conformal PID} & \multicolumn{3}{c}{Ours} \\\cmidrule(lr){2-4} \cmidrule(lr){5-7} \cmidrule(lr){8-10}
    &$\bar{J}_{\text{task}}~(\!\times\!10^{-3})$ & ${J}~(\!\times\!10^{-2})$ & $V$ & $\bar{J}_{\text{task}}~(\!\times\!10^{-3})$ & ${J}~(\!\times\!10^{-2})$ & $V$ & $\bar{J}_{\text{task}}~(\!\times\!10^{-3})$ & ${J}~(\!\times\!10^{-2})$ & $V$\\ \midrule
    0.5 & $-2.17$ & $2.081$ & $23.37\%$ & $-2.17$ & $2.0026$ & $22.77\%$ & $-2.02$ & $0.410$ & $4.85\%$   \\
    0.4 & $-2.15$ & $1.675$ & $18.94\%$ & $-2.15$ & $1.625$  & $18.30\%$ & $-1.98$ & $0.327$ & $4.21\%$   \\
    0.3 & $-2.12$ & $1.260$ & $14.67\%$ & $-2.12$ & $1.209$ & $13.91\%$ & $-1.92$ & $0.244$ & $3.29\%$   \\
    0.2 & $-2.10$ & $0.815$ & $10.45\%$ & $-2.11$ & $0.749$ & $9.81\%$ & $-1.83$ & $0.162$ & $2.41\%$\\
    0.1 & $-2.06$ & $0.376$ & $5.98\%$ & $-2.05$ & $0.341$ & $5.90\%$ & $-1.91$ & $0.079$ & $1.53\%$\\
    0.01 & $-1.94$ & $0.023$ & $0.89\%$ & $-1.94$ & $0.012$ & $0.71\%$ & $-1.88$ & $0.007$ & $0.30\%$ \\
\bottomrule
\end{tabular}}
\label{tab:results}
\end{table*}

Fig. \ref{fig:traj} shows excerpts of system trajectories for VanderPol under standard ACI and our method with $\alpha=0.5$. When the system dynamics vary over time, both methods initially exit the admissible set. Our method accounts for violation severity and quickly recovers after a single violation, whereas standard ACI ignores severity, leading to slower recovery and ten violations over the same period.


\section{Conclusion and Limitations}
\label{sec:conclusion}
In this work, we introduced Cost-Aware ACI, a framework that overcomes the severity-blindness of standard ACI by integrating violation cost directly into a boosted loss function, and proposed a closed-loop method for online risk-aware controller synthesis in unknown and dynamic environments. Our theoretical analysis proved that this method provides simultaneous, non-asymptotic upper bounds on both long-run constraint violation frequency and average cumulative violation cost. Experiments show that our method substantially reduces both violation frequency and cost compared with standard ACI, with only a modest trade-off in task performance. This work provides a principled, verifiable way to manage not only the rate of failures but also severity, moving toward more meaningful runtime assurance for learning-based control. However, our method still has several limitations. It does not support complex temporal logic specifications, and similar to standard ACI, may lead to less smooth system behavior. Addressing these issues remains part of future work.



 \bibliographystyle{plain}
 \bibliography{reference}

\clearpage
\section*{Appendix}
\setcounter{secnumdepth}{2}
\setcounter{section}{0}
\renewcommand\thesubsection{\thesection.\arabic{subsection}}
\renewcommand\thesection{\Alph{section}}
\setcounter{lemma}{0}
\setcounter{theorem}{0}
\setcounter{corollary}{0}
\setcounter{proposition}{1}

\section{Proof}
\label{sec:app_proof}
\subsection{Proof of Lemma~\ref{lemma:boundedness}}
\begin{lemma}[Parameter Boundedness]
Let $\{\delta_k\}$ be generated by the update rule \eqref{eq:update_rule}. Under Assumptions \ref{assumption:cost} and \ref{assumption:feasibility}, if $\delta_1$ is initialized within the interval $[-\gamma(L_{\max} - \alpha), 1 + \gamma\alpha]$, then $\delta_k$ remains within this interval for all $k \in \mathbb{N}$.
\end{lemma}

\begin{proof}
\textbf{Lower bound:} We prove the lower bound by contradiction. Assume there is a first time $k_0$ such that $\delta_{k_0} < -\gamma(L_{\max} - \alpha)$. This implies $\delta_{k_0-1} \geq -\gamma(L_{\max} - \alpha)$ and that the loss $L_{k_0-1}$ must have been large enough to cause the drop. Specifically, $L_{k_0-1} > \alpha$.
\begin{enumerate}
\item If $\delta_{k_0-1} < 0$, then by definition $\hat{Q}_{k_0-1}(\delta_{k_0-1}) = M$. By Assumption \ref{assumption:feasibility}, the controller enforces $\widehat{h}(\bm{x}_{k_0-2}, \bm{u}_{k_0-2}) \ge M$, which implies $s_{k_0-1} \le M = \hat{Q}_{k_0-1}$, leading to $e_{k_0-1} = 0$ and $L_{k_0-1} = 0$. This contradicts $L_{k_0-1} > \alpha$. 
\item If $\delta_{k_0-1} \ge 0$, the minimum possible $\delta_{k_0}$ is $0 + \gamma(\alpha - L_{\max}) = -\gamma(L_{\max} - \alpha)$, establishing the contradiction.
\end{enumerate}

\noindent \textbf{Upper bound:} The argument follows symmetrically. If $\delta_{k_0} > 1 + \gamma\alpha$, consider the  minimal such $k_0$. For $\delta_{k_0-1} > 1$, we have $\hat{Q}_{k_0-1} = -\epsilon$ (since $1-\delta_{k_0-1} < 0$), forcing $e_{k_0} = 1$ and $L_{k_0} > 1 > \alpha$, which decreases $\delta_{k_0}$, a contradiction. For $\delta_{k_0-1} \le 1$, the maximum possible $\delta_{k_0}$ is $1 + \gamma\alpha$.
\end{proof}

\subsection{Proof of Lemma~\ref{lemma:convergence}}
\begin{lemma}[Long-Run Average Boosted Loss Convergence]
Under the same conditions as Lemma \ref{lemma:boundedness}, the long-run time-average of the boosted loss converges to the target level $\alpha$, i.e., $\lim_{T \to \infty} \frac{1}{T} \sum_{k=1}^{T} L_k = \alpha$.
\end{lemma}

\begin{proof}
Expanding the recursion yields, we obtain $\delta_{T+1} = \delta_1 + \gamma \sum_{k=1}^{T} (\alpha - L_k)$. Rearranging and dividing by $T\gamma$, we obtain $\frac{1}{T} \sum_{k=1}^{T} L_k - \alpha = \frac{\delta_1- \delta_{T+1}}{T\gamma}$. From Lemma \ref{lemma:boundedness}, both $\delta_1$ and $\delta_{T+1}$ are bounded. Let $B = 1 + \gamma\alpha + \gamma(L_{\max} - \alpha)$ be the interval width. Then, we have
\begin{equation*}
    \left| \frac{1}{T} \sum_{k=1}^{T} L_k - \alpha \right| \le \frac{B}{T\gamma} \to 0 \quad \text{as } T \to \infty.
\end{equation*}
The proof is completed.
\end{proof}

\subsection{Proof of Theorem~\ref{thm:validity}}

\begin{theorem}[Upper Bound on Constraint Violation Frequency]
Let $v_k:=\mathbb{I}(h(\bm{x}_k) < 0)$ be the indicator of a constraint violation at time $k$. Under Assumptions \ref{assumption0} and \ref{assumption:cost}, the  violation frequency $V$ is upper-bounded by $\alpha$:
\begin{equation*}
V = \limsup_{T \to \infty} \frac{1}{T} \sum_{k=1}^{T} v_k \le \alpha.
\end{equation*}
\end{theorem}
\begin{proof}
A constraint violation at time $k$ ($v_k=1$, so $h(\bm{x}_k)<0$) implies a miscoverage event ($e_k=1$). This is because the controller enforces $\widehat{h}(\bm{x}_{k-1}, \bm{u}_{k-1}) \geq \hat{Q}_k \ge 0$. Thus, a violation implies:
\[
s_k = |h(\bm{x}_k) - \widehat{h}(\bm{x}_{k-1}, \bm{u}_{k-1})| = \widehat{h}(\bm{x}_{k-1}, \bm{u}_{k-1}) - h(\bm{x}_k) > \hat{Q}_k - 0 = \hat{Q}_k.
\]
This means $v_k=1 \implies e_k=1$, so $v_k \le e_k$. Furthermore, by definition, $L_k = e_k(1+\beta c(\bm{x}_k)) \ge e_k$. Combining these inequalities, we have $\frac{1}{T} \sum_{k=1}^{T} v_k \le \frac{1}{T} \sum_{k=1}^{T} e_k \le \frac{1}{T} \sum_{k=1}^{T} L_k$.
Taking the limit superior as $T \to \infty$ and applying Lemma \ref{lemma:convergence}, we obtain $\limsup_{T \to \infty} \frac{1}{T} \sum_{k=1}^{T} v_k \le \limsup_{T \to \infty} \frac{1}{T} \sum_{k=1}^{T} L_k = \alpha$.
The proof is completed. 
\end{proof}

\subsection{Proof of Theorem~\ref{thm:cost-bound}}

\begin{theorem}[Upper Bounds on Average Cumulative Cost $J$]
Under Assumptions \ref{assumption0} and \ref{assumption:cost}, for any risk sensitivity parameter $\beta > 0$, the long-run average cumulative constraint violation cost $J$ satisfies:
\begin{equation*}
    J= \limsup_{T \to \infty} \frac{1}{T} \sum_{k=1}^{T}  c(\bm{x}_k) \le \frac{\alpha}{\beta}.
\end{equation*}
\end{theorem}
\begin{proof}
From Lemma \ref{lemma:convergence} and the definition of $L_k$, we have $\alpha = \lim_{T \to \infty} \frac{1}{T} \sum_{k=1}^{T} L_k 
    = \lim_{T \to \infty} \left( \frac{1}{T} \sum_{k=1}^{T} e_k + \beta \frac{1}{T} \sum_{k=1}^{T} e_k c(\bm{x}_k) \right)$. Since both terms are nonnegative and $c(\bm{x}_k)=0$ when $e_k=0$, we have
    $\beta J\leq \beta \limsup_{T \to \infty} \frac{1}{T} \sum_{k=1}^{T} e_k c(\bm{x}_k) \le \alpha$.
Dividing both sides by $\beta$, we obtain 
\[
\limsup_{T \to \infty} \frac{1}{T} \sum_{k=1}^{T} e_k c(\bm{x}_k) \le \frac{\alpha}{\beta}.
\]
Since $c(\bm{x}_k)=0$ whenever $e_k=0$, it follows that
\[
\limsup_{T \to \infty} \frac{1}{T} \sum_{k=1}^{T} c(\bm{x}_k) \le \frac{\alpha}{\beta}.
\]
This completes the proof.
\end{proof}

\subsection{Proof of Corollary~\ref{coro:violation}}

\begin{corollary}[Finite-Time Violation Frequency and Cost Bounds]
Let $V_T:=\frac{1}{T}\sum_{k=1}^T v_k$ be the violation frequency and $J_T:=\frac{1}{T}\sum_{k=1}^T c(\bm{x}_k)$ be the average cumulative cost over the first $T$ time steps. The following bounds hold:
\begin{enumerate}
\item \textbf{Violation Frequency:} $V_T \le \alpha + \frac{\delta_1 - \delta_{T+1}}{T\gamma} = \alpha + O(1/T)$.
\item \textbf{Average Cost:} $J_T \le \frac{\alpha}{\beta} + \frac{\delta_1 - \delta_{T+1}}{\beta T\gamma} - \frac{1}{\beta T} \sum_{k=1}^T e_k \le \frac{\alpha}{\beta} + O(1/T)$.
\end{enumerate}
\end{corollary}
\begin{proof}
The proof follows directly from the relationships $v_k \le e_k \le L_k$ and $c(\bm{x}_k) \le e_k c(\bm{x}_k)$. Summing the core equality $\frac{1}{T}\sum L_k = \alpha + \frac{\delta_1 - \delta_{T+1}}{T\gamma}$ and rearranging terms yield the desired bounds. The $O(1/T)$ rates are a consequence of the boundedness of $\delta_k$ established in Lemma \ref{lemma:boundedness}.
\end{proof}

\subsection{Proof of Proposition~\ref{thm:conservation}}
\begin{proposition}[Conservation of Risk Budget]
The algorithm dynamically trades off miscoverage frequency against violation severity. Let $\hat{\alpha}_T:=\frac{1}{T}\sum_{k=1}^T e_k$ be the empirical miscoverage rate and $\bar{c}_T:= \frac{\sum_{k=1}^T e_k c(\bm{x}_k)}{\sum_{k=1}^T e_k}$ be the average cost conditional on a miscoverage event. In the long run, these quantities adhere to the conservation law:
$\lim_{T \to \infty} \hat{\alpha}_T (1 + \beta \bar{c}_T) = \alpha$.
\end{proposition}
\begin{proof}
Decomposing the sum of the boosted loss gives $\sum_{k=1}^{T} L_k = \sum_{k=1}^{T} e_k + \beta \sum_{k=1}^{T} e_k c(\bm{x}_k)$. We can factor out the term $\sum e_k$:
\[
\sum_{k=1}^{T} L_k = \left(\sum_{k=1}^{T} e_k\right) \left(1 + \beta \frac{\sum_{k=1}^{T} e_k c(\bm{x}_k)}{\sum_{k=1}^{T} e_k}\right) = T\hat{\alpha}_T (1 + \beta \bar{c}_T).
\]
Dividing by $T$ and applying Lemma \ref{lemma:convergence} yield the result.
\end{proof}

\subsection{Proof of Proposition~\ref{thm:regret}}
\begin{proposition}[Regret Minimization Perspective]
The update rule $\delta_{k+1}=\delta_k+\gamma (\alpha-L_k)$ is an instance of Online Gradient Descent (OGD) on the sequence of linearized loss functions $\mathcal{L}_k(\delta) = (L_k - \alpha)\delta$. Consequently, the algorithm achieves sublinear regret against the best fixed parameter $\delta^\star$ in hindsight:
\begin{equation*}
    \text{Regret}_T = \sum_{k=1}^{T} \mathcal{L}_k(\delta_{k}) - \min_{\delta^\star \in [0,1]} \sum_{k=1}^{T} \mathcal{L}_k(\delta^\star) \le O(\sqrt{T}).
\end{equation*}
\end{proposition}

\begin{proof}
The gradient of the loss is $\nabla \mathcal{L}_k(\delta) = L_k - \alpha$. The OGD update is $\delta_{k+1} = \delta_k - \gamma \nabla \mathcal{L}_k(\delta_k) = \delta_k - \gamma(L_k - \alpha)$, which matches our update rule exactly. The $O(\sqrt{T})$ regret bound is a standard result for OGD on a sequence of bounded linear losses.
\end{proof}

\section{Practical Implementation Strategies}
\label{sec:app_alg}
In this subsection, we provide further intuition and implementation details for Algorithm~\ref{alg:risk_sensitive_ACI}. We begin by outlining its main workflow. Algorithm~\ref{alg:risk_sensitive_ACI} operates in a continuous loop:

1. \textbf{Feedback and Adaptation ($k>0$)}: Upon observing state $\bm{x}_k$, the algorithm computes the nonconformity score $s_k$ and the cost-aware loss $L_k$. This loss is used to update the adaptive parameter to $\delta_{k+1}$.

2. \textbf{Calibration}: Using the updated history of scores $\mathcal{C}$ and the new parameter $\delta_{k+1}$, the algorithm re-computes the assurance threshold $\hat{Q}_{k+1}$.

3. \textbf{Admissible Control Synthesis}: The algorithm solves the optimization problem \eqref{eq:robust_constraint} using the new threshold $\hat{Q}_{k+1}$ to determine the optimal, admissible control action $\bm{u}_k^*$. This action is then applied to the system.

Next, we discuss several practical implementation strategies. To ensure real-time feasibility in high-frequency control loops, two modifications can be applied to Alg. \ref{alg:risk_sensitive_ACI}:

\textbf{Sparse Model Updates:} The surrogate models $\widehat{h}$ and $\widehat{\bm{f}}$ need not be updated at every time step $k$. In practice, they are often updated via \textit{batch learning} (e.g., every $N$ steps) or \textit{event-triggered learning} (e.g., only when prediction error exceeds a threshold) to reduce computational overhead without significantly degrading prediction accuracy. 

\textbf{Sliding Window Calibration:} To ensure adaptivity and computational tractability, the calibration set $\mathcal{C}$ is implemented as a \textit{sliding window} of fixed size $W$. This method bounds the complexity of the quantile computation and, more importantly, enhances responsiveness to non-stationarity. As the surrogate model $\widehat{h}$ is updated online, the window naturally discards ``stale'' residuals from prior model versions. This ensures the uncertainty threshold $\hat{Q}_k$ adapts to local distribution shifts and always reflects the performance of the \textit{current} model.


\section{Additional Experimental Results and Implementation Details}
\label{sec:app_exp}

All computations are performed on a machine equipped with an Intel i7-13700H 2.10 GHz CPU and 32 GB of RAM. The MPC optimization problem is solved using CasADi \cite{andersson2018casadi} in conjunction with the Ipopt solver \cite{wachter2006implementation}. The 10 random seeds used in Fig.~\ref{fig:results} are $0,1,\ldots,9$. For all other experiments, the random seed is fixed to 0.

\subsection{Implementation Details}
\textbf{Vanderpol.} Consider the VanderPol oscillator adapted from \cite{henrion2013convex},
\begin{equation*}
    \begin{cases}
    x_{k+1}=x_k + 0.01(-2y_k + u_{k} + p_k),\\
    y_{k+1}=y_k + 0.01(0.8x_k-10(x_k^2-0.21)y_k+v_{k} + q_k),
    \end{cases}
\end{equation*}
where the state space $\mathcal{X} = \{(x,y)\mid 2-x^2-y^2\geq 0\}$, the admissible set $\mathcal{S} = \{(x,y)\mid 1 - x^2 - y^2\geq 0\}$, and the control set $\mathcal{U} = \{(u,v)\mid -40 \leq u,v\leq 40\}$. The parameter variation $p_{k}$ and $q_{k}$ are modeled as piecewise-constant, time-varying signals. Every 10 time steps, each parameter variation is independently resampled from $14\times \text{Beta}(0.1,0.1)-7$ and kept constant until the next update. This setting captures slowly varying dynamics in the system model. In this benchmark, the MPC task objective $J_\text{task}$ encourages the system state to remain as close as possible to $(0.8, 0.8)$, which directly conflicts with the constraints.

\noindent\textbf{Pendulum.} Consider the \textit{Pendulum} benchmark adapted from GYM \cite{towers2024gymnasium},
\begin{equation*}
    \begin{cases}
    \dot{\theta}_{k+1}=\dot{\theta}_{k} + 0.05(15\sin{(\theta_{k})}  + 3u_{k} + r_{k} + q_{k}),\\
    \theta_{k+1}=\theta_{k} + 0.05\dot{\theta}_{k+1},
    \end{cases}
\end{equation*}
where the state space $\mathcal{X} = \{(\dot{\theta},\theta)\mid 2-(\frac{\dot{\theta}}{6})^4-(\frac{\theta}{1.5\pi})^4\geq 0\}$, the admissible set $\mathcal{S} = \{(\dot{\theta},\theta)\mid 1-(\frac{\dot{\theta}}{6})^4-(\frac{\theta}{1.5\pi})^4\geq 0\}$, and the control set $\mathcal{U} = \{u\mid -80 \leq u\leq 80\}$. The terms $r_k$ and $q_k$ jointly model time-varying effects in the system dynamics. Specifically, $r_{k}$ represents a piecewise-constant parameter variation: every 50 time steps, $r_{k}$ is independently resampled from a uniform distribution over $[-5,5]$ and kept constant until the next update. In addition, $q_{k}$ captures a drifting bias, which increases by $0.01$ every $50$ time steps, with $q_{0} = 0$. The MPC task objective $J_\text{task}$ encourages the magnitude of the state $\dot{\theta}_{k}$ to be as large as possible.

\noindent\textbf{MountainCar.} Consider the \textit{MountainCarContinuous} benchmark adapted from GYM \cite{towers2024gymnasium},
\begin{equation*}
    \begin{cases}
    v_{k+1}=v_{k} + (0.0015 + m_{k})u_{k} - 0.0025\cos{(3p_{k})} + r_{k},\\
    p_{k+1}=p_{k} + v_{k+1} + q_{k},
    \end{cases}
\end{equation*}
where the state space $\mathcal{X} = \{(v,p)\mid 2-(\frac{v}{0.05})^4-(\frac{p+0.3}{0.7})^4\geq 0\}$, the admissible set $\mathcal{S} = \{(v,p)\mid 1-(\frac{v}{0.05})^4-(\frac{p+0.3}{0.7})^4\geq 0\}$, and the control set $\mathcal{U} = \{u\mid -200 \leq u\leq 200\}$. Here, $m_{k}$ gradually decreases over time. Specifically, every five time steps, $m_{k}$ decreases by $5\times 10^{-8}$. In addition, every five time steps, the parameter variations $r_{k}$ and $q_{k}$ are independently resampled from the distribution $0.004\times \text{Beta}(0.1,0.1) -0.002$ and then held constant until the next update. The MPC task objective $J_\text{task}$ encourages the magnitude of the state $v_{k}$ to be as large as possible.

\noindent\textbf{Lorenz.} Consider a 6-dimensional Lorenz model adapted from \cite{lorenz1996predictability},
\begin{equation*}
    \begin{cases}
    \bm{x}_{k+1}[0]=\bm{x}_{k}[0] + 0.01((\bm{x}_{k}[1] - \bm{x}_{k}[4])\bm{x}_{k}[5] - \bm{x}_{k}[0] + \bm{u}_{k}[0] + p_k),\\
    \bm{x}_{k+1}[1]=\bm{x}_{k}[1] + 0.01((\bm{x}_{k}[2] - \bm{x}_{k}[5])\bm{x}_{k}[0] - \bm{x}_{k}[1] + \bm{u}_{k}[1]),\\
    \bm{x}_{k+1}[2]=\bm{x}_{k}[2] + 0.01((\bm{x}_{k}[3] - \bm{x}_{k}[0])\bm{x}_{k}[1] - \bm{x}_{k}[2] + \bm{u}_{k}[2]),\\
    \bm{x}_{k+1}[3]=\bm{x}_{k}[3] + 0.01((\bm{x}_{k}[4] - \bm{x}_{k}[1])\bm{x}_{k}[2] - \bm{x}_{k}[3] + \bm{u}_{k}[3]),\\
    \bm{x}_{k+1}[4]=\bm{x}_{k}[4] + 0.01((\bm{x}_{k}[5] - \bm{x}_{k}[2])\bm{x}_{k}[3] - \bm{x}_{k}[4] + \bm{u}_{k}[4]),\\
    \bm{x}_{k+1}[5]=\bm{x}_{k}[5] + 0.01((\bm{x}_{k}[0] - \bm{x}_{k}[3])\bm{x}_{k}[4] - \bm{x}_{k}[5] + \bm{u}_{k}[5]),
\end{cases}
\end{equation*}
where $\bm{x}$ and $\bm{u}$ are six-dimensional vectors, and $\bm{x}_k[i]$, for $i = 0,1,\ldots,5$, denotes the $i$-th component of the state $\bm{x}$ at time step $k$. The state space is $\mathcal{X} = \{\bm{x}\mid 12-\sum_{i=0}^5 \bm{x}^2[i]\geq 0\}$, the admissible set $\mathcal{S} = \{\bm{x}\mid 1-\sum_{i=0}^5 \bm{x}^2[i]\geq 0\}$, and the control set $\mathcal{U} = \{\bm{u}\mid \bigwedge_{i=0}^5 -100 \leq \bm{u}[i]\leq 100\}$. Specifically, every five time steps, the parameter variations $p_{k}$ is independently resampled from a uniform distribution over $[-5,5]$ and kept constant until the next update. The MPC task objective $J_\text{task}$ encourages the value of the state $\bm{x}[0]$ to be as small as possible.

\paragraph{Implementation Details:}
We employ a physics-informed offline learning procedure to initialize $\widehat{h}$ and $\widehat{\bm f}$. Specifically, we obtain the initial surrogate models by interpolation-based fitting over a compact dictionary of standard basis functions, such as polynomial and trigonometric features. These models serve as warm-start surrogates in Algorithm~1. To account for unmodeled time-varying effects, we further augment the nominal predictors with lightweight residual corrections learned from online one-step prediction errors \cite{ljung1983theory}. The residual components are updated only when the corresponding prediction error exceeds a prescribed threshold, following the sparse online update strategy described in Appendix~B. For example, in the Vanderpol benchmark, $\hat{h}$ is updated when $|h-\widehat{h}| > 0.01$.

In the experiments reported in Figure~\ref{fig:results} and Table~\ref{tab:results}, the ACI learning rate is set to $\gamma=0.01$, and the quantile estimation uses a sliding window of the 200 most recent scores (see Appendix~B). The sensitivity parameter $\beta$ is set to $50$ for Vanderpol and $100$ for Pendulum, MountainCar, and Lorenz. Across all experiments, the violation cost is defined as $\rho(\bm{x}) = -\frac{h(\bm{x})}{M_h}$, where $M_h$ normalizes the cost to satisfy Assumption \ref{assumption:cost}. The horizon in the MPC optimization \eqref{eq:robust_constraint} is set to $N=20$.

\subsection{Sensitivity Analysis}
\label{sec:app_sensitivity}

In this section, we further analyze the algorithm's sensitivity to its key hyperparameters, $\gamma$, $\beta$, and the sliding window size $W$, to verify the robustness of the performance gains and to confirm that these parameters behave as theoretically expected. We analyze each parameter in a representative setting where its influence is most evident for clarity: $\gamma$ and $\beta$ on MountainCar with $\alpha =0.1$, and the sliding window size on Pendulum with $\alpha=0.3$.
\begin{table*}[!h]
\centering
\caption{Performance comparison under different values of $\gamma$}
\small
\setlength{\tabcolsep}{1.5mm}{
    \begin{tabular}{*{7}{c}}
    \toprule
    \multirow{2}{*}{$\gamma$} & \multicolumn{3}{c}{Standard ACI} & \multicolumn{3}{c}{Ours} \\\cmidrule(lr){2-4} \cmidrule(lr){5-7}
    & $\bar{J}_{\text{task}}$ & ${J}$ & $V$ & $\bar{J}_{\text{task}}$ & ${J}$ & $V$\\ \midrule
    0.2 & $-2.03\times10^{-3}$ & $5.00\!\times\!10^{-3}$ & $4.92\%$ & $-1.84\times10^{-3}$ & $0.83\!\times\!10^{-3}$ & $0.84\%$  \\
    0.1 & $-2.04\times10^{-3}$ & $4.58\!\times\!10^{-3}$ & $5.45\%$ & $-1.85\times10^{-3}$ & $0.83\!\times\!10^{-3}$ & $1.05\%$  \\
    0.01 & $-2.06\times10^{-3}$ & $3.76\!\times\!10^{-3}$ & $5.98\%$ & $-1.91\times10^{-3}$ & $0.79\!\times\!10^{-3}$ & $1.53\%$  \\
    0.001 & $-2.05\times10^{-3}$ & $3.53\!\times\!10^{-3}$ & $6.11\%$ & $-1.98\times10^{-3}$ & $0.79\!\times\!10^{-3}$ & $2.38\%$  \\
\bottomrule
\end{tabular}}
\label{tab:gamma}
\end{table*}
\begin{figure}[h]
    \centering
\subfigure[Performance Variation Under Different Values of $\beta$]{    
    \includegraphics[width=0.46\linewidth]{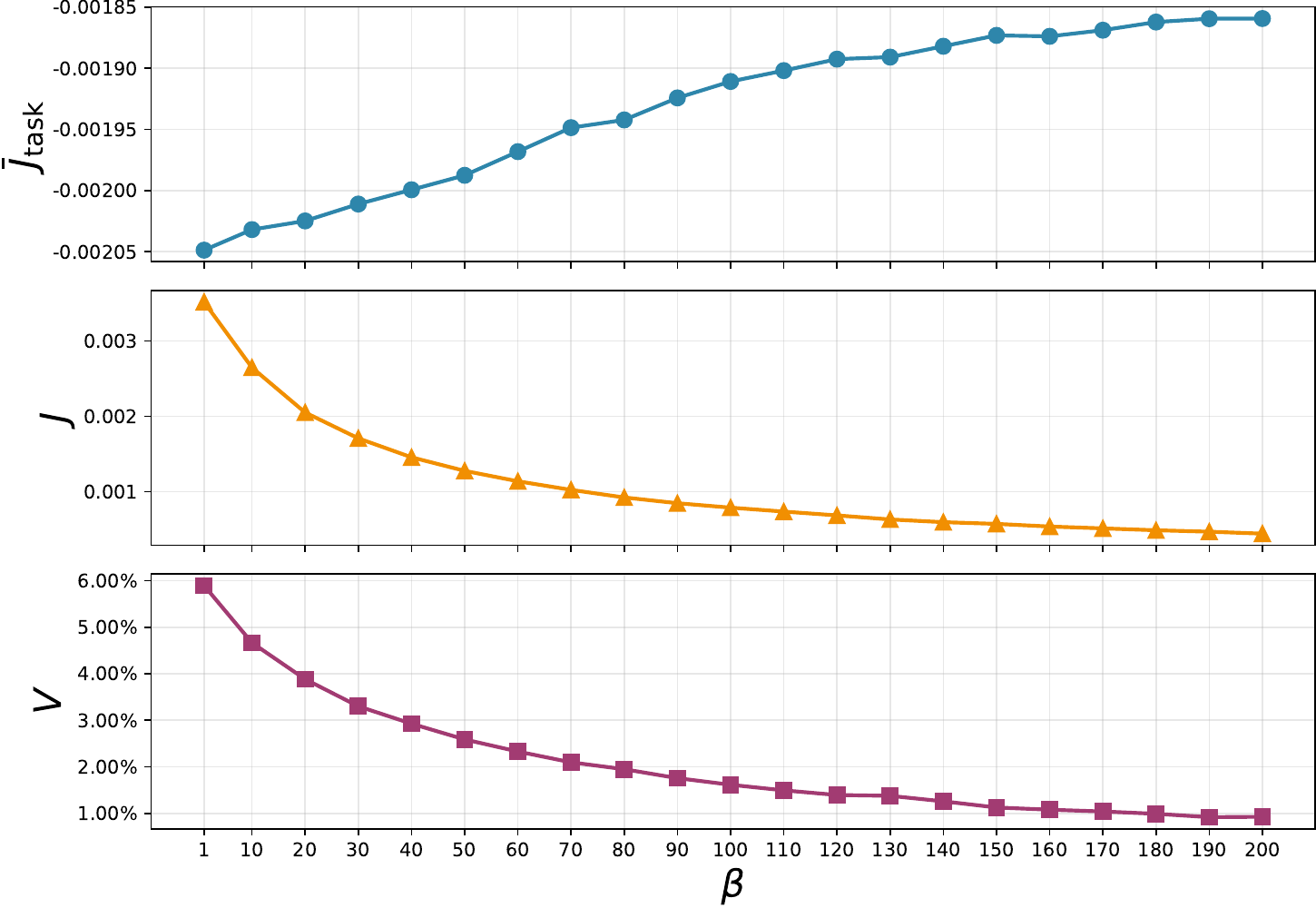}
    \label{fig:beta}
    }
    \subfigure[Comparison of $\hat{Q}_{k+1}$ Under Different Sliding Window Size $W$]{ 
    \includegraphics[width=0.46\linewidth]{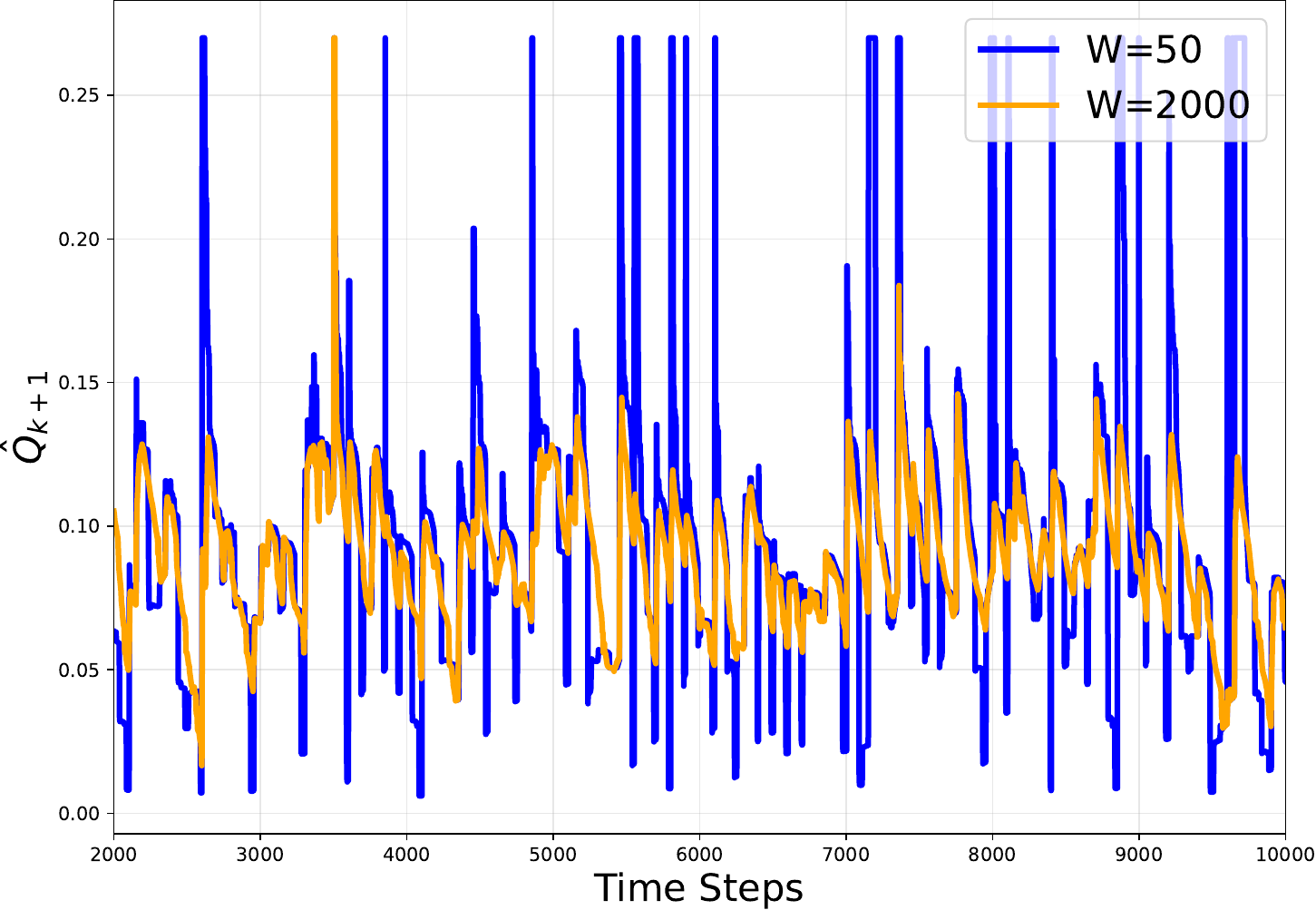}
    \label{fig:size}
    }
    \caption{Effect of Sensitivity $\beta$ and $W$}
\end{figure}


\noindent \textbf{Effect of Learning Rate $\gamma$:} Table \ref{tab:gamma} shows that our method's superiority over standard ACI is robust across different learning rates. For all tested values of $\gamma$, our approach maintains a significant advantage in reducing both cost $J$ and frequency $V$. The choice of $\gamma$ itself involves a trade-off related to the system's temporal dynamics, and as noted in \cite{gibbs2024conformal}, it can also be adapted online. We will investigate this in the future work.

\noindent \textbf{Effect of Sensitivity $\beta$:} Fig. \ref{fig:beta} shows that the parameter $\beta$ acts as a tuning knob for runtime assurance. As $\beta$ increases, the algorithm penalizes violations more heavily, making it more conservative. This directly improves assurance, causing both the average cost $J$ and violation frequency $V$ to decrease. In exchange, the task cost $J_{task}$ slightly increases, which is expected since our benchmarks are designed to oppose assurance. This confirms that $\beta$ provides a clear and predictable way to trade task performance for a higher level of assurance.


\noindent \textbf{Effect of the sliding window size:} 
The sliding window size, $W$, is a trade-off between stability and responsiveness. As shown in Fig. \ref{fig:size}, a small window ($W=50$) causes a volatile quantile $\hat{Q}_{k+1}$ that closely tracks recent events, while a large window ($W=2000$) provides a much smoother, more stable estimate. Despite this difference in behavior, Table \ref{tab:size_window} shows that our method's superior assurance performance is robust. Across all tested window sizes, our approach consistently outperforms standard ACI by achieving lower violation costs ($J$) and violation frequencies ($V$). This is a key practical advantage, as it shows our algorithm's effectiveness does not depend on carefully tuning this parameter.
\begin{table*}[!h]
\centering
\caption{Performance comparison under different size of the sliding window}
\small
\setlength{\tabcolsep}{1.5mm}{
    \begin{tabular}{*{7}{c}}
    \toprule
    \multirow{2}{*}{W} & \multicolumn{3}{c}{Standard ACI} & \multicolumn{3}{c}{Ours} \\\cmidrule(lr){2-4} \cmidrule(lr){5-7}
    & $\bar{J}_{\text{task}}$ & ${J}$ & $V$ & $\bar{J}_{\text{task}}$ & ${J}$ & $V$\\ \midrule
    50 & $-28.93$ & $2.71\!\times\!10^{-3}$ & $14.54\%$ & $-28.82$ & $1.09\!\times\!10^{-3}$ & $9.28\%$  \\
    500 & $-28.82$ & $3.82\!\times\!10^{-3}$ & $14.89\%$ & $-28.85$ & $1.30\!\times\!10^{-3}$ & $8.23\%$  \\
    2000 & $-28.86$ & $3.67\!\times\!10^{-3}$ & $14.43\%$ & $-28.81$ & $1.23\!\times\!10^{-3}$ & $8.11\%$  \\
\bottomrule
\end{tabular}}
\label{tab:size_window}
\end{table*}

\end{document}